\documentclass[a4paper,UKenglish]{lipics}
\usepackage{amsmath}
\usepackage{semantic}
\usepackage{hyperref}
\usepackage[normalem]{ulem}
\usepackage{graphicx}
\usepackage{caption}
\usepackage{todonotes}
\usepackage{microtype}
\usepackage{enumerate}

\setpremisesend{0pt} 

\newcommand{\mytodo}[2][]{}
\newcommand{\redtodo}[2][]{}

\makeatletter
\@ifundefined{lhs2tex.lhs2tex.sty.read}%
  {\@namedef{lhs2tex.lhs2tex.sty.read}{}%
   \newcommand\SkipToFmtEnd{}%
   \newcommand\EndFmtInput{}%
   \long\def\SkipToFmtEnd#1\EndFmtInput{}%
  }\SkipToFmtEnd

\newcommand\ReadOnlyOnce[1]{\@ifundefined{#1}{\@namedef{#1}{}}\SkipToFmtEnd}
\usepackage{amstext}
\usepackage{amssymb}
\usepackage{stmaryrd}
\DeclareFontFamily{OT1}{cmtex}{}
\DeclareFontShape{OT1}{cmtex}{m}{n}
  {<5><6><7><8>cmtex8
   <9>cmtex9
   <10><10.95><12><14.4><17.28><20.74><24.88>cmtex10}{}
\DeclareFontShape{OT1}{cmtex}{m}{it}
  {<-> ssub * cmtt/m/it}{}

\DeclareFontShape{OT1}{cmtt}{bx}{n}
  {<5><6><7><8>cmtt8
   <9>cmbtt9
   <10><10.95><12><14.4><17.28><20.74><24.88>cmbtt10}{}
\DeclareFontShape{OT1}{cmtex}{bx}{n}
  {<-> ssub * cmtt/bx/n}{}

\newcommand{\Conid}[1]{\mathit{#1}}
\newcommand{\Varid}[1]{\mathit{#1}}
\newcommand{\anonymous}{\kern0.06em \vbox{\hrule\@width.5em}}
\newcommand{\plus}{\mathbin{+\!\!\!+}}

\renewcommand{\leq}{\leqslant}

\usepackage{polytable}

\@ifundefined{mathindent}%
  {\newdimen\mathindent\mathindent\leftmargini}%
  {}%

\def\resethooks{%
  \global\let\SaveRestoreHook\empty
  \global\let\ColumnHook\empty}
\newcommand*{\savecolumns}[1][default]%
  {\g@addto@macro\SaveRestoreHook{\savecolumns[#1]}}
\newcommand*{\restorecolumns}[1][default]%
  {\g@addto@macro\SaveRestoreHook{\restorecolumns[#1]}}
\newcommand*{\aligncolumn}[2]%
  {\g@addto@macro\ColumnHook{\column{#1}{#2}}}

\resethooks

\newcommand{\onelinecommentchars}{\quad-{}- }
\newcommand{\commentbeginchars}{\enskip\{-}
\newcommand{\commentendchars}{-\}\enskip}

\newcommand{\visiblecomments}{%
  \let\onelinecomment=\onelinecommentchars
  \let\commentbegin=\commentbeginchars
  \let\commentend=\commentendchars}

\newcommand{\invisiblecomments}{%
  \let\onelinecomment=\empty
  \let\commentbegin=\empty
  \let\commentend=\empty}

\visiblecomments

\newlength{\blanklineskip}
\newcommand{\hsindent}[1]{\quad}
\let\hspre\empty
\let\hspost\empty

\EndFmtInput
\makeatother
\ReadOnlyOnce{polycode.fmt}%
\makeatletter

\newcommand{\hsnewpar}[1]%
  {{\parskip=0pt\parindent=0pt\par\vskip #1\noindent}}

\newcommand{\hscodestyle}{}

\newcommand{\sethscode}[1]%
  {\expandafter\let\expandafter\hscode\csname #1\endcsname
   \expandafter\let\expandafter\endhscode\csname end#1\endcsname}

  {\par\noindent
   \advance\leftskip\mathindent
   \hscodestyle
   \let\\=\@normalcr
   \let\hspre\(\let\hspost\)%
   \pboxed}%
  {\endpboxed\)%
   \par\noindent
   \ignorespacesafterend}

  {\hsnewpar\abovedisplayskip
   \advance\leftskip\mathindent
   \hscodestyle
   \let\hspre\(\let\hspost\)%
   \pboxed}%
  {\endpboxed%
   \hsnewpar\belowdisplayskip
   \ignorespacesafterend}

  {\hsnewpar\abovedisplayskip
   \advance\leftskip\mathindent
   \hscodestyle
   \let\\=\@normalcr
   \(\pboxed}%
  {\endpboxed\)%
   \hsnewpar\belowdisplayskip
   \ignorespacesafterend}

\newcommand{\plainhs}{\sethscode{plainhscode}}

\plainhs

  {\hsnewpar\abovedisplayskip
   \advance\leftskip\mathindent
   \hscodestyle
   \let\\=\@normalcr
   \(\parray}%
  {\endparray\)%
   \hsnewpar\belowdisplayskip
   \ignorespacesafterend}

  {\parray}{\endparray}

  {\(\parray}{\endparray\)}

\def\codeframewidth{\arrayrulewidth}
\RequirePackage{calc}

  {\parskip=\abovedisplayskip\par\noindent
   \hscodestyle
   \arrayrulewidth=\codeframewidth
   \tabular{@{}|p{\linewidth-2\arraycolsep-2\arrayrulewidth-2pt}|@{}}%
   \hline\framedhslinecorrect\\{-1.5ex}%
   \let\endoflinesave=\\
   \let\\=\@normalcr
   \(\pboxed}%
  {\endpboxed\)%
   \framedhslinecorrect\endoflinesave{.5ex}\hline
   \endtabular
   \parskip=\belowdisplayskip\par\noindent
   \ignorespacesafterend}

\newcommand{\framedhslinecorrect}[2]%
  {#1[#2]}

  {\(\def\column##1##2{}%
   \let\>\undefined\let\<\undefined\let\\\undefined
   \newcommand\>[1][]{}\newcommand\<[1][]{}\newcommand\\[1][]{}%
   \def\fromto##1##2##3{##3}%
   }{\) }%

  {\let\orighscode=\hscode
   \let\origendhscode=\endhscode
   \def\endhscode{\def\hscode{\endgroup\def\@currenvir{hscode}\\}\begingroup}
   \orighscode\def\hscode{\endgroup\def\@currenvir{hscode}}}%
  {\origendhscode
   \global\let\hscode=\orighscode
   \global\let\endhscode=\origendhscode}%

\makeatother
\EndFmtInput
\newcommand{\der}{\,\vdash}
\newcommand{\of}{{:}}
\newcommand{\dom}{\mathop{\mathsf{dom}}\nolimits}
\newcommand{\Set}{\mathsf{Set}}
\newcommand{\C}{\mathcal{C}}
\newcommand{\jinf}{\Rightarrow}
\newcommand{\jchk}{\Leftarrow}

\newenvironment{caselist}{\begin{list}{{\it Case}}{\setlength{\leftmargin}{2ex}\setlength{\itemindent}{0ex}}}{\end{list}}

\newcommand{\nextcase}{\item}

\newcommand{\mgoal}[1][]{\mbox{goal\ifthenelse{\equal{#1}{}}{}{~#1}}}

\def\lv{\mathopen{{[\kern-0.14em[}}}    
\def\rv{\mathclose{{]\kern-0.14em]}}}   
\newcommand{\den}[1]{\lv #1 \rv}

\title{Type checking through unification}
\author[1]{Francesco Mazzoli}
\affil[1]{FP Complete\footnote{The work presented in this paper was performed at the Chalmers University of Technology.} \\ \href{mailto:f@mazzo.li}{\texttt{<f@mazzo.li>}}}
\author[2]{Andreas Abel}
\affil[2]{Gothenburg University \\ \href{mailto:andreas.abel@gu.se}{\texttt{<andreas.abel@gu.se>}}}
\authorrunning{F. Mazzoli and A. Abel}

\Copyright{Francesco Mazzoli and Andreas Abel}


\serieslogo{}

\begin{document}

{\Large Draft of 2016-09-30}

\vspace{5em}

{\let\newpage\relax\maketitle}

\begin{abstract}
  In this paper we describe how to leverage higher-order unification to
  type check a dependently typed language with meta-variables. The
  literature usually presents the unification algorithm as a standalone
  component, however the need to check definitional equality of terms
  while type checking gives rise to a tight interplay between type
  checking and unification.  This interplay is a major source of
  complexity in the type-checking algorithm for existing dependently
  typed programming languages.  We propose an algorithm that encodes a
  type-checking problem entirely in the form of unification constraints,
  reducing the complexity of the type-checking code by taking advantage
  of higher order unification, which is already part of the
  implementation of many dependently typed languages.
\end{abstract}

\section{Introduction}
\label{introduction}

Theories with dependent types have been successfully exploited to design
programming languages and theorem provers, such as Agda
\cite{norell2007}, Idris \cite{brady2013}, or Coq \cite{coqart}.  To
make these systems practical, the user is presented with a language much
richer than the underlying type theory, which will hopefully be small
enough to gain confidence in the correctness of the code that type
checks it.

One common way to make a type theory palatable is extending it with
\emph{meta-variables}, standing for yet to be determined terms, and
solved by unification.  Their usage in traditional programming languages
is confined to type inference, and thus traditionally they can stand for
types only.  In dependently typed languages types can contain terms,
and thus meta-variables are usually extended to stand for any term in
our language.  A typical use case for meta-variables is \emph{implicit
  arguments} as introduced by Pollack \cite{pollack90}, relieving the
user of having to write easily inferrable arguments to functions.
For example, in
Agda we can write a safe \ensuremath{\Varid{head}} function which extracts the first
element of a list, inferring both the type of the elements and the
length of the list:
\begin{hscode}\SaveRestoreHook
\column{B}{@{}>{\hspre}l<{\hspost}@{}}%
\column{3}{@{}>{\hspre}l<{\hspost}@{}}%
\column{E}{@{}>{\hspre}l<{\hspost}@{}}%
\>[3]{}\Varid{head}\mathbin{:}\{\mskip1.5mu \Conid{A}\mathbin{:}\mathsf{Set} \mskip1.5mu\}\to \{\mskip1.5mu \Varid{n}\mathbin{:}\mathsf{Nat} \mskip1.5mu\}\to \mathsf{Vec} \;\Conid{A}\;(\mathrm{1}\mathbin{+}\Varid{n})\to \Conid{A}{}\<[E]%
\\
\>[3]{}\Varid{head}\;(\Varid{x}\mathbin{::}\Varid{xs})\mathrel{=}\Varid{x}{}\<[E]%
\ColumnHook
\end{hscode}\resethooks
Here, \ensuremath{\mathsf{Vec} \;\Conid{A}\;\Varid{n}} denotes a list of length \ensuremath{\Varid{n}} with elements of type \ensuremath{\Conid{A}},
and \ensuremath{\mathsf{Set} } is the type of types.  The expression \ensuremath{\{\mskip1.5mu \Conid{A}\mathbin{:}\mathsf{Set} \mskip1.5mu\}\to \{\mskip1.5mu \Varid{n}\mathbin{:}\mathsf{Nat} \mskip1.5mu\}\to \mathbin{...}}
binds two implicit arguments.  When invoking \ensuremath{\Varid{head}}, the type checker
will insert two meta-variables standing for \ensuremath{\Conid{A}} and \ensuremath{\Varid{n}} and
attempt to solve them by inspecting the \ensuremath{\mathsf{Vec} } argument that follows.
Note that \ensuremath{\Varid{n}} is a value, while in languages such ML and Haskell only
types can be implicit.

The task of integrating meta-variables in a type-checking algorithm for
dependent types gives rise to complications.  For example, consider the
task of type checking
\begin{hscode}\SaveRestoreHook
\column{B}{@{}>{\hspre}l<{\hspost}@{}}%
\column{3}{@{}>{\hspre}l<{\hspost}@{}}%
\column{E}{@{}>{\hspre}l<{\hspost}@{}}%
\>[3]{}\mathsf{true} \mathbin{:}\mathbf{if}\;\alpha \leq \mathrm{2}\;\mathbf{then}\;\mathsf{Bool} \;\mathbf{else}\;\mathsf{Nat} ,{}\<[E]%
\ColumnHook
\end{hscode}\resethooks
where \ensuremath{\alpha } is a yet to be determined (\emph{uninstantiated})
meta-variable of type \ensuremath{\mathsf{Nat} }.  We want the type of \ensuremath{\mathsf{true} } to be \ensuremath{\mathsf{Bool} },
but reduction is impeded by \ensuremath{\alpha }.  Thus, we cannot complete type
checking until \ensuremath{\alpha } is instantiated.\footnote{Note that we cannot
  instantiate \ensuremath{\alpha } without loss of generality, since both \ensuremath{\mathrm{0}} and \ensuremath{\mathrm{1}}
  are acceptable solutions.}  The problem lies in the fact that type
checking dependent types involves reducing terms to their normal forms,
something that can be affected by meta-variables, like in this
case.

To solve issues like the one above, the only viable option---apart from
refusing to solve them---is to wait for the meta-variables that are
affecting type checking to be instantiated, and then resume.  This gives
rise to a sort of concurrency that makes reasoning about the type
checking algorithm arduous.  In this paper, expanding on ideas developed
in Agda \cite{norell2007} and Epigram \cite{mcbride2004}, we propose an
algorithm that encodes a type-checking problem in a set of unification
constraints with a single traversal of the term to be checked.  The
generated constraints can be solved effectively by the unification
procedure already employed by Agda, but our elaboration procedure is
considerably simpler and shorter than Agda's type-checking code.  This
highlights an overlap in functionality between the type checker, which
needs to check that types and terms are of a certain shape; and the
unifier, which checks the equality of terms.  Moreover, our algorithm
lets us clearly separate concerns between type checking and unification,
making it easier to gain confidence on the elaboration procedure and
then experiment with various unification ``backends''.

In the rest of the paper, we will explain the problem more clearly
(Section~\ref{problem}).  Then we will introduce a simple type theory
(Section~\ref{type-theory}) that will serve as a vector to explain our
algorithm in detail.  In Section~\ref{unification} we will give a
specification to the unification procedure.  The algorithm itself is
presented in Section~\ref{algorithm}, along with some of its
properties. We will then briefly discuss the performance and how the
algorithm can be extended to support certain popular language features
(Section~\ref{remarks}).


We have implemented the presented algorithm in a prototype,
\texttt{tog}, which covers a subset of Agda---every \texttt{tog} program
is also a valid Agda program.\footnote{The source code for \texttt{tog}
  is available at \url{https://github.com/bitonic/tog}.}

\section{The problem}
\label{problem}

In this section we will explain the challenges faced
when type checking dependent types with meta-variables.  An Agda-like
syntax will be used throughout the examples, please refer to
Appendix~\ref{examples-syntax} for clarifications.

Coming back to the problem of type checking
\begin{hscode}\SaveRestoreHook
\column{B}{@{}>{\hspre}l<{\hspost}@{}}%
\column{3}{@{}>{\hspre}l<{\hspost}@{}}%
\column{E}{@{}>{\hspre}l<{\hspost}@{}}%
\>[3]{}\mathsf{true} \mathbin{:}\Conid{BoolOrNat}\;\alpha ,{}\<[E]%
\ColumnHook
\end{hscode}\resethooks
given unistantiated meta-variable \ensuremath{\alpha } and definition
\begin{hscode}\SaveRestoreHook
\column{B}{@{}>{\hspre}l<{\hspost}@{}}%
\column{3}{@{}>{\hspre}l<{\hspost}@{}}%
\column{E}{@{}>{\hspre}l<{\hspost}@{}}%
\>[3]{}\Conid{BoolOrNat}\mathbin{:}\mathsf{Nat} \to \mathsf{Set} {}\<[E]%
\\
\>[3]{}\Conid{BoolOrNat}\mathrel{=}\lambda \Varid{x}\to \mathbf{if}\;\Varid{x}\leq \mathrm{2}\;\mathbf{then}\;\mathsf{Bool} \;\mathbf{else}\;\mathsf{Nat} {}\<[E]%
\ColumnHook
\end{hscode}\resethooks
there are various tempting ways to approach the problem.  The most
conservative approach is to stop type checking when faced with
\emph{blocked} terms (terms whose normalization is impeded by some
meta-variables).  However, this approach is unsatisfactory in many
instances.  Consider
\begin{hscode}\SaveRestoreHook
\column{B}{@{}>{\hspre}l<{\hspost}@{}}%
\column{3}{@{}>{\hspre}l<{\hspost}@{}}%
\column{E}{@{}>{\hspre}l<{\hspost}@{}}%
\>[3]{}(\mathsf{true} ,\mathsf{refl} )\mathbin{:}\Conid{BoolOrNat}\;\alpha \times (\alpha \equiv \mathrm{0}){}\<[E]%
\ColumnHook
\end{hscode}\resethooks
Where \ensuremath{\Varid{x}\equiv \Varid{y}} is the type inhabited by proofs that \ensuremath{\Varid{x}} is equal to \ensuremath{\Varid{y}}
(propositional equality), and \ensuremath{\mathsf{refl} } is of type \ensuremath{\Varid{t}\equiv \Varid{t}} for any \ensuremath{\Varid{t}}
(reflexivity).  Type checking this pair will involve type checking \ensuremath{\mathsf{true} \mathbin{:}\Conid{BoolOrNat}\;\alpha } and then \ensuremath{\mathsf{refl} \mathbin{:}\alpha \equiv \mathrm{0}}. If we give up on the
first type-checking problem, we will not examine the second, which will
give us a solution for \ensuremath{\alpha } (\ensuremath{\alpha \mathbin{:=}\mathrm{0}}).  After instantiating
\ensuremath{\alpha } we can easily go back and successfully type check the first
part.  In general, we want to attempt to type check as much as possible,
and to instantiate as many meta-variables as possible---as long as we do
so without loss of generality, like in this case.

Another approach is to assume that blocked type-checking problems will
eventually be solved, and continue type checking.  However, this road is
dangerous since we need to be careful not to generate ill-typed terms or
invalid type-checking contexts, as noted by
Norell and Coquand \cite{norell2007b}. Consider
\begin{hscode}\SaveRestoreHook
\column{B}{@{}>{\hspre}l<{\hspost}@{}}%
\column{3}{@{}>{\hspre}l<{\hspost}@{}}%
\column{E}{@{}>{\hspre}l<{\hspost}@{}}%
\>[3]{}\Varid{test}\mathbin{:}(\alpha \equiv \mathrm{0})\times (((\Varid{x}\mathbin{:}\Conid{BoolOrNat}\;\alpha )\to \Conid{BoolOrNat}\;(\mathrm{1}\mathbin{+}\Varid{x}))\to \mathsf{Nat} ){}\<[E]%
\\
\>[3]{}\Varid{test}\mathrel{=}(\mathsf{refl} ,\lambda \Varid{g}\to \Varid{g}\;\mathsf{true} ){}\<[E]%
\ColumnHook
\end{hscode}\resethooks
Type checking the definition \ensuremath{\Varid{test}} will involve checking that its type
is a valid type, and that its body is well typed.  Checking the former
will involve making sure that
\begin{hscode}\SaveRestoreHook
\column{B}{@{}>{\hspre}l<{\hspost}@{}}%
\column{3}{@{}>{\hspre}l<{\hspost}@{}}%
\column{E}{@{}>{\hspre}l<{\hspost}@{}}%
\>[3]{}\Conid{BoolOrNat}\;\alpha \mathrel{=}\mathsf{Nat} {}\<[E]%
\ColumnHook
\end{hscode}\resethooks
since we know that the type of \ensuremath{\Varid{x}} must be \ensuremath{\mathsf{Nat} }, given that \ensuremath{\Varid{x}} is
used as an argument of \ensuremath{(\mathrm{1}\mathbin{+})\mathbin{:}\mathsf{Nat} \to \mathsf{Nat} }.\footnote{Note that checking
  that an equality type is a well-formed type does not involved checking
  that the equated things are equal---\ensuremath{\mathrm{4}\equiv \mathrm{5}} is a valid type.  In this
  instance while \ensuremath{\alpha \equiv \mathrm{0}} appears in the type for \ensuremath{\Varid{test}}, this does
  not mean that \ensuremath{\alpha } will be unified with \ensuremath{\mathrm{0}} when type checking the
  type.  However, type checking its proof \ensuremath{\mathsf{refl} \mathbin{:}\alpha \equiv \mathrm{0}} will.}

If we assume that the type is valid, we will proceed and type check
the body pairwise.  Type checking the first element---a proof by
reflexivity that \ensuremath{\alpha } is equal to \ensuremath{\mathrm{0}}---will instantiate \ensuremath{\alpha }
to \ensuremath{\mathrm{0}}, and then we will be faced with
\begin{hscode}\SaveRestoreHook
\column{B}{@{}>{\hspre}l<{\hspost}@{}}%
\column{3}{@{}>{\hspre}l<{\hspost}@{}}%
\column{E}{@{}>{\hspre}l<{\hspost}@{}}%
\>[3]{}(\lambda \Varid{g}\to \Varid{g}\;\mathsf{true} )\mathbin{:}((\Varid{x}\mathbin{:}\mathsf{Bool} )\to \Conid{BoolOrNat}\;(\mathrm{1}\mathbin{+}\Varid{x}))\to \mathsf{Nat} {}\<[E]%
\ColumnHook
\end{hscode}\resethooks
Note that the type is ill-typed,\footnote{\ensuremath{\Varid{x}}, of type \ensuremath{\mathsf{Bool} }, appears
  as an argument to the function \ensuremath{(\mathrm{1}\mathbin{+})}.} violating the usual
invariants present when type checking---namely the fact that when we
make progress we always generate well-typed terms.  Worse, to type check
we will instantiate \ensuremath{\Varid{x}} with \ensuremath{\mathrm{0}}, ending up with \ensuremath{\Conid{BoolOrNat}\;(\mathrm{1}\mathbin{+}\mathsf{true} )}.
With some effort we can exploit this problem to make the type checker
loop\mytodo{How? Perhaps such an example would be too long for inclusion
  in the paper, but the current text doesn't convince me that you (or
  someone else) have actually constructed such an example.}, and thus
type checking will be undecidable.

As mentioned in the introduction, at the heart of the problem lies the
fact that to type check we need to reduce terms to their weak head
normal form. If reduction is impeded by meta-variables, we cannot
proceed.  To overcome this problem, Norell proposed to define type
checking as an \emph{elaboration} procedure:\footnote{Norell was
  following a long tradition of elaborating user syntax when working
  with type theories, see Section~\ref{related-work} for more details.}
given the problem of type checking \ensuremath{\Varid{t}} against \ensuremath{\Conid{A}} in context \ensuremath{\Gamma },
type checking will produce a term \ensuremath{\Varid{u}} that approximates \ensuremath{\Varid{t}}:
\begin{hscode}\SaveRestoreHook
\column{B}{@{}>{\hspre}l<{\hspost}@{}}%
\column{3}{@{}>{\hspre}l<{\hspost}@{}}%
\column{E}{@{}>{\hspre}l<{\hspost}@{}}%
\>[3]{}\llbracket \Gamma \vdash \Varid{t}\mathbin{:}\Conid{A}\rrbracket \leadsto \Varid{u}{}\<[E]%
\ColumnHook
\end{hscode}\resethooks
\ensuremath{\Varid{u}} is an approximation of \ensuremath{\Varid{t}} in the sense that it it can be turned
into \ensuremath{\Varid{t}} by instantiating certain meta-variables---if a subterm of \ensuremath{\Varid{t}}
cannot be type checked a placeholder meta-variable will be put in its
place, an type checking that subterm will be postponed.  Type checking
will also consist in making sure that, once the postponed type-checking
problems can be solved, the placeholder meta-variables will be
instantiated accordingly with the corresponding omitted subterm of \ensuremath{\Varid{t}}
(possibly instantiated further).

For instance, when type checking the type of \ensuremath{\Varid{test}}, we'll have
\begin{hscode}\SaveRestoreHook
\column{B}{@{}>{\hspre}l<{\hspost}@{}}%
\column{3}{@{}>{\hspre}l<{\hspost}@{}}%
\column{5}{@{}>{\hspre}l<{\hspost}@{}}%
\column{14}{@{}>{\hspre}l<{\hspost}@{}}%
\column{50}{@{}>{\hspre}l<{\hspost}@{}}%
\column{67}{@{}>{\hspre}l<{\hspost}@{}}%
\column{73}{@{}>{\hspre}l<{\hspost}@{}}%
\column{E}{@{}>{\hspre}l<{\hspost}@{}}%
\>[3]{}\llbracket \cdot \vdash {}\<[14]%
\>[14]{}((\Varid{x}\mathbin{:}\Conid{BoolOrNat}\;\alpha )\to \Conid{BoolOrNat}\;(\mathrm{1}\mathbin{+}\Varid{x}))\to \mathsf{Nat} {}\<[67]%
\>[67]{}\;\mathbin{:}{}\<[73]%
\>[73]{}\;\mathsf{Set} \rrbracket \leadsto {}\<[E]%
\\
\>[3]{}\hsindent{2}{}\<[5]%
\>[5]{}((\Varid{x}\mathbin{:}\Conid{BoolOrNat}\;\alpha )\to \Conid{BoolOrNat}\;\beta ){}\<[50]%
\>[50]{}\to \mathsf{Nat} {}\<[E]%
\ColumnHook
\end{hscode}\resethooks
Since we cannot type check
\begin{hscode}\SaveRestoreHook
\column{B}{@{}>{\hspre}l<{\hspost}@{}}%
\column{3}{@{}>{\hspre}l<{\hspost}@{}}%
\column{E}{@{}>{\hspre}l<{\hspost}@{}}%
\>[3]{}\Varid{x}\mathbin{:}\Conid{BoolOrNat}\;\alpha \vdash \mathrm{1}\mathbin{+}\Varid{x}\mathbin{:}\mathsf{Nat} {}\<[E]%
\ColumnHook
\end{hscode}\resethooks
a fresh meta-variable \ensuremath{\beta } of type \ensuremath{\mathsf{Nat} } in context \ensuremath{\Varid{x}\mathbin{:}\Conid{BoolOrNat}\;\alpha } replaces \ensuremath{\mathrm{1}\mathbin{+}\Varid{x}}.  Then, when checking the body
of \ensuremath{\Varid{test}}, we will check it against the approximated type generated
above.  When \ensuremath{\alpha } is instantiated, we can resume checking that
\ensuremath{\Conid{BoolOrNat}\;\alpha \mathrel{=}\mathsf{Nat} }, and if we are successful, instantiate \ensuremath{\beta \mathbin{:=}\mathrm{1}\mathbin{+}\Varid{x}}.  This will prevent us from running into problems when type
checking the body, since when we do instantiate \ensuremath{\alpha } to \ensuremath{\mathrm{0}}, we do
not have \ensuremath{\mathrm{1}\mathbin{+}\Varid{x}} later: instead, \ensuremath{\beta } is in its place, preserving the
well-typedness of the type.

The Agda system, as described in Norell's thesis,
currently implements this elaboration interleaving type
checking and unification, using some fairly delicate machinery.  Our
contribution is to describe a type-checking problem entirely in terms of
unification constraints, thus simplifying the algorithm.  This
highlights an overlap in functionality between the type checker, which
needs to check that types and terms are of a certain shape, and the
unifier, which checks the equality of terms: we are using the unifier as
an engine to pattern match on types, with taking meta-variables into
account. Moreover, separating the unifier from the type checker makes it
easy easy to experiment with different unification ``backends'' used by
the same type checking ``frontend''. \mytodo{It would be nice to
  highlight that this makes it easier to maintain the invariants we want
  in the type checker.}

\section{The type theory}
\label{type-theory}

To illustrate the type-checking algorithm we will make use of a dependent
type theory with booleans and one universe.
Its syntax is shown in Figure~\ref{syntax}.

\subsection{Terms and types}

Terms and types inhabit the same syntactic class.  We usually denote
terms with \ensuremath{\Varid{t}}, \ensuremath{\Varid{u}}, and \ensuremath{\Varid{v}}; and types with \ensuremath{\Conid{A}}, \ensuremath{\Conid{B}}, and \ensuremath{\Conid{C}}.  The
theory is designed to be the simplest fragment that presents the
problems described in Section~\ref{problem}.  For this reason we include
a universe \ensuremath{\mathsf{Set} } and means of computing with booleans, so that we can
write functions from booleans to types---otherwise meta-variables can
never prevent us from knowing how a type looks like.  The typing rules
and algorithms presented in this paper can be extended to a richer theory,
as we have done for our implementation, which includes implicit
arguments, user defined inductive data types and records, and
propositional equality.

\begin{figure}
  \begin{minipage}[b]{.58\textwidth}
    \begin{hscode}\SaveRestoreHook
\column{B}{@{}>{\hspre}l<{\hspost}@{}}%
\column{7}{@{}>{\hspre}l<{\hspost}@{}}%
\column{21}{@{}>{\hspre}l<{\hspost}@{}}%
\column{26}{@{}>{\hspre}l<{\hspost}@{}}%
\column{71}{@{}>{\hspre}l<{\hspost}@{}}%
\column{E}{@{}>{\hspre}l<{\hspost}@{}}%
\>[7]{}\Varid{x},\Varid{y},\Varid{z}{}\<[71]%
\>[71]{}\mbox{\onelinecomment  Variables}{}\<[E]%
\\
\>[7]{}\alpha ,\beta ,\gamma {}\<[71]%
\>[71]{}\mbox{\onelinecomment  Meta-variables}{}\<[E]%
\\[\blanklineskip]%
\>[7]{}\mbox{\onelinecomment  Types/terms}{}\<[E]%
\\
\>[7]{}\Conid{A},\Conid{B},\Conid{C},\Varid{t},\Varid{u},\Varid{v}{}\<[E]%
\\
\>[7]{}\hsindent{14}{}\<[21]%
\>[21]{}\mathbin{::=}{}\<[26]%
\>[26]{}\mathsf{Set} {}\<[71]%
\>[71]{}\mbox{\onelinecomment  Type of types}{}\<[E]%
\\
\>[7]{}\hsindent{14}{}\<[21]%
\>[21]{}\mid {}\<[26]%
\>[26]{}\mathsf{Bool} \mid \mathsf{true} \mid \mathsf{false} {}\<[71]%
\>[71]{}\mbox{\onelinecomment  Booleans}{}\<[E]%
\\
\>[7]{}\hsindent{14}{}\<[21]%
\>[21]{}\mid {}\<[26]%
\>[26]{}(\Varid{x}\mathbin{:}\Conid{A})\to \Conid{B}\mid \lambda \Varid{x}\to \Varid{t}{}\<[71]%
\>[71]{}\mbox{\onelinecomment  Dependent functions}{}\<[E]%
\\
\>[7]{}\hsindent{14}{}\<[21]%
\>[21]{}\mid {}\<[26]%
\>[26]{}\Varid{n}{}\<[71]%
\>[71]{}\mbox{\onelinecomment  Neutral term}{}\<[E]%
\\[\blanklineskip]%
\>[7]{}\mbox{\onelinecomment  Neutral terms}{}\<[E]%
\\
\>[7]{}\Varid{n}{}\<[21]%
\>[21]{}\mathbin{::=}{}\<[26]%
\>[26]{}\Varid{h}{}\<[71]%
\>[71]{}\mbox{\onelinecomment  Head}{}\<[E]%
\\
\>[21]{}\mid {}\<[26]%
\>[26]{}\Varid{n}\;\Varid{t}{}\<[71]%
\>[71]{}\mbox{\onelinecomment  Function application}{}\<[E]%
\\
\>[21]{}\mid {}\<[26]%
\>[26]{}\mathbf{if}\;\Varid{n}\;/\Varid{x}.\Conid{A}\;\mathbf{then}\;\Varid{t}\;\mathbf{else}\;\Varid{u}{}\<[71]%
\>[71]{}\mbox{\onelinecomment  \ensuremath{\mathsf{Bool} } elimination}{}\<[E]%
\\[\blanklineskip]%
\>[7]{}\Varid{h}{}\<[21]%
\>[21]{}\mathbin{::=}{}\<[26]%
\>[26]{}\Varid{x}\mid \alpha {}\<[71]%
\>[71]{}\mbox{\onelinecomment  Neutral term heads}{}\<[E]%
\\[\blanklineskip]%
\>[7]{}\Gamma ,\Delta {}\<[21]%
\>[21]{}\mathbin{::=}\cdot \mid \Gamma ,\Varid{x}\mathbin{:}\Conid{A}{}\<[71]%
\>[71]{}\mbox{\onelinecomment  Contexts}{}\<[E]%
\\[\blanklineskip]%
\>[7]{}\Sigma ,\Xi {}\<[21]%
\>[21]{}\mathbin{::=}\cdot \mid \Sigma ,\alpha \mathbin{:}\Conid{A}{}\<[71]%
\>[71]{}\mbox{\onelinecomment  Signatures}{}\<[E]%
\\[\blanklineskip]%
\>[7]{}\theta ,\eta {}\<[21]%
\>[21]{}\mathbin{::=}\cdot \mid \theta ,\alpha \mathbin{:=}\Varid{t}{}\<[71]%
\>[71]{}\mbox{\onelinecomment  Meta-variable substitutions}{}\<[E]%
\ColumnHook
\end{hscode}\resethooks
    \caption{Syntax}
    \label{syntax}
  \end{minipage}
\end{figure}

\subsection{Neutral terms, substitutions, and term application}

Terms are always kept in \ensuremath{\beta }-normal form. Terms headed by a variable
or meta-variable are called \emph{neutral}, the others \emph{canonical}.
Variable and meta-variable substitution immediately restore the
\ensuremath{\beta }-normal form as soon as the substitution is performed, a technique
known as \emph{hereditary substitution} \cite{watkins2003}.
Substituting a term \ensuremath{\Varid{u}} for a head \ensuremath{\Varid{h}}, i.e., a variable \ensuremath{\Varid{x}} or
meta-variable \ensuremath{\alpha }, is written \ensuremath{\Varid{t} [\mskip1.5mu \Varid{h}\mathbin{:=}\Varid{u}\mskip1.5mu]}, reading ``substitute
\ensuremath{\Varid{u}} for \ensuremath{\Varid{h}} in \ensuremath{\Varid{t}}''. The rules for substitution are not relevant to the
paper and are shown in Appendix~\ref{substitution}.  The appendix also
defines rules to eliminate redexes that substitution might generate, to
restore the \ensuremath{\beta }-normal form.

\subsection{Contexts and signatures}

Most operations are done under a context (denoted by \ensuremath{\Gamma } or
\ensuremath{\Delta }), that stores the types of free variables; and a signature
(denoted by \ensuremath{\Sigma } or \ensuremath{\Xi }), that stores the type of meta-variables.  We
tacitly assume that no duplicate names are present in contexts and
signatures.  We often make use of a global signature \ensuremath{\Sigma } throghout the
rules, if there is no need for the rules to carry it explicitely since
it is never changed.  Note that a signature contains only closed types
for the meta-variables---we do not make use of an explicit
representation of meta-variables in context. This is for the sake of
simplicity, since we do not present our unification algorithm in detail,
where the contextual representation would be most convenient. Throughout
the paper, we will use \ensuremath{\Gamma \to \Conid{A}} to indicate the function type formed
by all the types in \ensuremath{\Gamma } as the domains and terminating with \ensuremath{\Conid{A}}, and
\ensuremath{\Varid{t}\;\Gamma } to indicate the term formed by \ensuremath{\Varid{t}} applied to all the
variables in \ensuremath{\Gamma }.  Moreover, every mention of a context and a
signature is assumed to be valid according to the rules in
Figure~\ref{context-validity} and \ref{signature-validity}.

Throughout the paper, we will also concatenate whole contexts and
signatures, e.g. \ensuremath{\Sigma ,\Xi }.

\subsection{Well-typed terms}

The bidirectional typing checking rules are shown in figures
\ref{rules-type-checking} and \ref{rules-type-inference}.  The type of
neutral terms can be inferred, while canonical terms are checked.  Our
type theory includes a universe \ensuremath{\mathsf{Set} } equipped with an inconsistent
typing rule \ensuremath{\mathsf{Set} \mathbin{:}\mathsf{Set} } for the sake of simplicity, but our presentation
can be extended with stratified universes.

Finally, definitional equality of terms (needed to define the typing
rules) is specified in figures \ref{term-conversion} and
\ref{spine-conversion}.  The conversion rules are performed in a
type-directed way, so that it can respect the \ensuremath{\eta }-laws of
functions.

\mytodo[inline]{Maybe mention some properties of the type theory?
  Normalization, decidability of type checking, etc?
  Andreas: These properties fail because of Set : Set}

\subsection{Meta-variable substitutions}

We specify typed meta-variable substitutions, a construct that will be
useful to give a specification to our unification algorithm.  A
meta-variable substitution \ensuremath{\theta } from \ensuremath{\Sigma } to \ensuremath{\Xi } gives an
instantiation to each meta-variable in \ensuremath{\Sigma }.  The meta-variables in the
instantiations are scoped over a new signature \ensuremath{\Xi }.  The validity rule
shown in figure \ref{well-formed-subst} makes sure that the
meta-variables instantiations are well-typed with respect to their type.

Applying a substitution \ensuremath{\theta } to a term \ensuremath{\Varid{t}} amounts to substitute each
meta-variable for its instantiation in \ensuremath{\theta }, as specified in appendix
\ref{substitution}.  Moreover, if the substitution is valid, if we have
\ensuremath{\Sigma ;\Gamma \vdash \Varid{t}\mathbin{:}\Conid{A}} and \ensuremath{\Xi \vdash \theta \mathbin{:}\Sigma }, we will have that
\begin{hscode}\SaveRestoreHook
\column{B}{@{}>{\hspre}l<{\hspost}@{}}%
\column{3}{@{}>{\hspre}l<{\hspost}@{}}%
\column{E}{@{}>{\hspre}l<{\hspost}@{}}%
\>[3]{}\Xi ;\theta  \Gamma \vdash \theta  \Varid{t}\mathbin{:}\theta  \Conid{A},{}\<[E]%
\ColumnHook
\end{hscode}\resethooks
where applying a substitution to a context amounts to applying it to all
the types it contains.

Like with contexts and signatures, we will use \ensuremath{,} to also concatenate
meta-variable substitutions.  So for example if we have \ensuremath{\Xi \vdash \theta \mathbin{:}\Sigma } and \ensuremath{\Xi \vdash \eta \mathbin{:}\Sigma' }, we have that \ensuremath{\Xi \vdash (\theta ,\eta )\mathbin{:}(\Sigma ,\Sigma' )}.

\begin{figure}
  \[
  \inference{}{\ensuremath{\Gamma \vdash \mathsf{Set} \mathbin{:}\mathsf{Set} }}\hfill
  \inference{}{\ensuremath{\Gamma \vdash \mathsf{Bool} \mathbin{:}\mathsf{Set} }}\hfill
  \inference{}{\ensuremath{\Gamma \vdash \mathsf{true} \mathbin{:}\mathsf{Bool} }}\hfill
  \inference{}{\ensuremath{\Gamma \vdash \mathsf{false} \mathbin{:}\mathsf{Bool} }}
  \]
  \[
  \inference{\ensuremath{\Gamma \vdash \Conid{A}\mathbin{:}\mathsf{Set} } & \ensuremath{\Gamma ,\Varid{x}\mathbin{:}\Conid{A}\vdash \Conid{B}\mathbin{:}\mathsf{Set} }}{
    \ensuremath{\Gamma \vdash (\Varid{x}\mathbin{:}\Conid{A})\to \Conid{B}\mathbin{:}\mathsf{Set} }
  }\hfill
  \inference{\ensuremath{\Gamma ,\Varid{x}\mathbin{:}\Conid{A}\vdash \Varid{t}\mathbin{:}\Conid{B}}}{\ensuremath{\Gamma \vdash \lambda \Varid{x}\to \Varid{t}\mathbin{:}(\Varid{x}\mathbin{:}\Conid{A})\to \Conid{B}}}
  \hfill
  \inference{\ensuremath{\Gamma \vdash \Varid{n}\Rightarrow \Conid{A}} & \ensuremath{\Gamma \vdash \Conid{A}\equiv \Conid{B}\mathbin{:}\mathsf{Set} }}{
    \ensuremath{\Gamma \vdash \Varid{n}\mathbin{:}\Conid{B}}
  }
  \]
  \caption{\boxed{\ensuremath{\Sigma ;\Gamma \vdash \Varid{t}\mathbin{:}\Conid{A}}} Type checking canonical terms}
  \label{rules-type-checking}
\end{figure}

\begin{figure}
    \[
    \inference{\ensuremath{\Varid{x}\mathbin{:}\Conid{A}\in \Gamma }}{
      \ensuremath{\Gamma \vdash \Varid{x}\Rightarrow \Conid{A}}
    }\hfill
    \inference{\ensuremath{\alpha \mathbin{:}\Conid{A}\in \Sigma }}{
      \ensuremath{\Gamma \vdash \alpha \Rightarrow \Conid{A}}
    }\hfill
    \inference{\ensuremath{\Gamma \vdash \Varid{n}\Rightarrow (\Varid{x}\mathbin{:}\Conid{A})\to \Conid{B}} & \ensuremath{\Gamma \vdash \Varid{t}\mathbin{:}\Conid{A}}}{
      \ensuremath{\Gamma \vdash \Varid{n}\;\Varid{t}\Rightarrow \Conid{B}[\Varid{x}\mathbin{:=}\Varid{t}]}
    }
    \]
    \[
    \inference{\ensuremath{\Gamma \vdash \Varid{n}\Rightarrow \mathsf{Bool} } & \ensuremath{\Gamma ,\Varid{x}\mathbin{:}\mathsf{Bool} \vdash \Conid{A}\mathbin{:}\mathsf{Set} } &
      \ensuremath{\Gamma \vdash \Varid{t}\mathbin{:}\Conid{A}[\Varid{x}\mathbin{:=}\mathsf{true} ]} &
      \ensuremath{\Gamma \vdash \Varid{u}\mathbin{:}\Conid{A}[\Varid{x}\mathbin{:=}\mathsf{false} ]} &
    }{
      \ensuremath{\Gamma \vdash \mathbf{if}\;\Varid{n}\;/\Varid{x}.\Conid{A}\;\mathbf{then}\;\Varid{t}\;\mathbf{else}\;\Varid{u}\Rightarrow \Conid{A}[\Varid{x}\mathbin{:=}\Varid{n}]}
    }
    \]
    \caption{\boxed{\ensuremath{\Sigma ;\Gamma \vdash \Varid{n}\Rightarrow \Conid{A}}} Type inference for neutral terms}
    \label{rules-type-inference}
\end{figure}

\begin{figure}
  \begin{minipage}[b]{0.48\textwidth}
    \[
    \inference{}{\ensuremath{\vdash \cdot }}\hfill
    \inference{
      \ensuremath{\vdash \Sigma } & \ensuremath{\Sigma ;\cdot \vdash \Conid{A}\mathbin{:}\mathsf{Set} }
    }{
      \ensuremath{\vdash \Sigma ,\alpha \mathbin{:}\Conid{A}}
    }
    \]
  \caption{\boxed{\ensuremath{\vdash \Sigma }} Well-formed signatures}
  \label{signature-validity}
  \end{minipage}
  \hfill
  \begin{minipage}[b]{.48\textwidth}
    \[
    \inference{}{\ensuremath{\Sigma \vdash \cdot }}\hfill
    \inference{
      \ensuremath{\Sigma \vdash \Gamma } & \ensuremath{\Sigma ;\Gamma \vdash \Conid{A}\mathbin{:}\mathsf{Set} }
    }{
      \ensuremath{\Sigma \vdash \Gamma ,\Varid{x}\mathbin{:}\Conid{A}}
    }
    \]
  \caption{\boxed{\ensuremath{\Sigma \vdash \Gamma }} Well-formed contexts}
  \label{context-validity}
  \end{minipage}
\end{figure}

\begin{figure}
    \[
    \inference{
      \text{for all \ensuremath{(\alpha \mathbin{:}\Conid{A})\in \Sigma }, \ensuremath{\alpha \mathbin{:=}\Varid{t}\in \theta } and \ensuremath{\;\Xi ;\cdot \vdash \Varid{t}\mathbin{:}\theta  \Conid{A}}}
    }{
      \ensuremath{\Xi \vdash \theta \mathbin{:}\Sigma }
    }
    \]
  \caption{\boxed{\ensuremath{\Xi \vdash \theta \mathbin{:}\Sigma }} Well-formed meta-variable substitutions}
  \label{well-formed-subst}
\end{figure}

\begin{figure}
  \[
  \inference{}{\ensuremath{\Gamma \vdash \mathsf{Set} \equiv \mathsf{Set} \mathbin{:}\mathsf{Set} }}\hfill
  \inference{}{\ensuremath{\Gamma \vdash \mathsf{Bool} \equiv \mathsf{Bool} \mathbin{:}\mathsf{Set} }}\hfill
  \inference{}{\ensuremath{\Gamma \vdash \mathsf{true} \equiv \mathsf{true} \mathbin{:}\mathsf{Bool} }}
  \]
  \[
  \inference{}{\ensuremath{\Gamma \vdash \mathsf{false} \equiv \mathsf{false} \mathbin{:}\mathsf{Bool} }}\hfill
  \inference{\ensuremath{\Gamma \vdash A_1 \equiv A_2 \mathbin{:}\mathsf{Set} } & \ensuremath{\Gamma ,\Varid{x}\mathbin{:}A_1 \vdash B_1 \equiv B_2 \mathbin{:}\mathsf{Set} }}{
    \ensuremath{\Gamma \vdash (\Varid{x}\mathbin{:}A_1 )\to B_1 \equiv (\Varid{x}\mathbin{:}A_2 )\to B_2 \mathbin{:}\mathsf{Set} }
  }
  \]
  \[
  \inference{\ensuremath{\Gamma ,\Varid{x}\mathbin{:}\Conid{A}\vdash \Varid{f}\;\Varid{x}\equiv \Varid{g}\;\Varid{x}\mathbin{:}\Conid{B}}}{\ensuremath{\Gamma \vdash \Varid{f}\equiv \Varid{g}\mathbin{:}(\Varid{x}\mathbin{:}\Conid{A})\to \Conid{B}}}
  \hfill
  \inference{
    \ensuremath{\Gamma \vdash \Varid{n}\equiv \Varid{n'}\Rightarrow \Conid{A}} &
    \ensuremath{\Gamma \vdash \Conid{A}\equiv \Conid{B}\mathbin{:}\mathsf{Set} }
  }{
    \ensuremath{\Gamma \vdash \Varid{n}\equiv \Varid{n'}\mathbin{:}\Conid{B}}
  }
  \]
  \caption{\boxed{\ensuremath{\Sigma ;\Gamma \vdash \Varid{t}\equiv \Varid{u}\mathbin{:}\Conid{A}}} Definitional equality of canonical terms}
  \label{term-conversion}
\end{figure}
\begin{figure}
  \[
  \inference{
    \ensuremath{\Gamma \vdash \Varid{h}\Rightarrow \Conid{A}}
  }{
    \ensuremath{\Gamma \vdash \Varid{h}\equiv \Varid{h}\Rightarrow \Conid{A}}
  }
  \hfill
  \inference{
    \ensuremath{\Gamma \vdash \Varid{n}\equiv \Varid{n'}\Rightarrow (\Varid{x}\mathbin{:}\Conid{A})\to \Conid{B}} &
    \ensuremath{\Gamma \vdash \Varid{t}\equiv \Varid{t'}\mathbin{:}\Conid{A}}
  }{
    \ensuremath{\Gamma \vdash \Varid{n}\;\Varid{t}\equiv \Varid{n'}\;\Varid{t'}\Rightarrow \Conid{B}[\Varid{x}\mathbin{:=}\Varid{n}]}
  }
  \]
  \[
  \inference{
    \ensuremath{\Gamma \vdash \Varid{n}\equiv \Varid{n'}\Rightarrow \mathsf{Bool} } \\
    \ensuremath{\Gamma ,\Varid{x}\mathbin{:}\mathsf{Bool} \vdash \Conid{A}\equiv \Conid{A'}\mathbin{:}\mathsf{Set} } &
    \ensuremath{\Gamma \vdash \Varid{t}\equiv \Varid{t'}\mathbin{:}\Conid{A}[\Varid{x}\mathbin{:=}\mathsf{true} ]} &
    \ensuremath{\Gamma \vdash \Varid{u}\equiv \Varid{u'}\mathbin{:}\Conid{A}[\Varid{x}\mathbin{:=}\mathsf{false} ]}
  }{
    \ensuremath{\Gamma \vdash \mathbf{if}\;\Varid{n}\;/\Varid{x}.\Conid{A}\;\mathbf{then}\;\Varid{t}\;\mathbf{else}\;\Varid{u}\equiv \mathbf{if}\;\Varid{n'}\;/\Varid{x}.\Conid{A'}\;\mathbf{then}\;\Varid{t'}\;\mathbf{else}\;\Varid{u'}\Rightarrow \Conid{A}[\Varid{x}\mathbin{:=}\Varid{n}]}
  }
  \]
  \caption{\boxed{\ensuremath{\Sigma ;\Gamma \vdash \Varid{n}\equiv \Varid{n'}\Rightarrow \Conid{A}}}
    Definitional equality of neutral terms}
  \label{spine-conversion}
\end{figure}

\section{Unification}
\label{unification}

In this section we will give a specification for the unification
algorithm that we will need to solve the constraints generated by
elaboration.

\subsection{Unification constraints}

The input for the unification algorithm are \emph{heterogeneous}
constraints of the form
\begin{hscode}\SaveRestoreHook
\column{B}{@{}>{\hspre}l<{\hspost}@{}}%
\column{3}{@{}>{\hspre}l<{\hspost}@{}}%
\column{E}{@{}>{\hspre}l<{\hspost}@{}}%
\>[3]{}\Gamma \vdash \Varid{t}\mathbin{:}\Conid{A}\mathrel{=}\Varid{u}\mathbin{:}\Conid{B}.{}\<[E]%
\ColumnHook
\end{hscode}\resethooks
Such a constraint is well formed if we have that \ensuremath{\Gamma \vdash \Varid{t}\mathbin{:}\Conid{A}} and
\ensuremath{\Gamma \vdash \Varid{u}\mathbin{:}\Conid{B}}.  As we will see in Section~\ref{algorithm}, it is
crucial for the constraints to be heterogeneous for the elaboration
procedure to work as we intend to.

A constraint is \emph{solved} if we have that \ensuremath{\Gamma \vdash \Conid{A}\equiv \Conid{B}\mathbin{:}\mathsf{Set} }
and \ensuremath{\Gamma \vdash \Varid{t}\equiv \Varid{u}\mathbin{:}\Conid{A}}.  This means that to solve a constraint the
unifier will have to establish definitional equality of both the types
and the terms.

\subsection{Unification algorithm specification}

A unification algorithm takes a signature and set of constraints, and
attempts to solve them by instantiating meta-variables.  Thus,
unification rules will be of the form
\begin{hscode}\SaveRestoreHook
\column{B}{@{}>{\hspre}l<{\hspost}@{}}%
\column{3}{@{}>{\hspre}l<{\hspost}@{}}%
\column{E}{@{}>{\hspre}l<{\hspost}@{}}%
\>[3]{}\Sigma ,\mathcal{C} \leadsto \Xi ,\theta {}\<[E]%
\ColumnHook
\end{hscode}\resethooks
Where \ensuremath{\Sigma } and \ensuremath{\mathcal{C} } are the input signature and constraints, and \ensuremath{\Xi }
and \ensuremath{\theta } are the output signature and substitution, such that \ensuremath{\Xi \vdash \theta \mathbin{:}\Sigma }.  If the substitution solves all the constraints we have
that
\begin{hscode}\SaveRestoreHook
\column{B}{@{}>{\hspre}l<{\hspost}@{}}%
\column{3}{@{}>{\hspre}l<{\hspost}@{}}%
\column{E}{@{}>{\hspre}l<{\hspost}@{}}%
\>[3]{}\Xi \Vdash \theta  \mathcal{C} .{}\<[E]%
\ColumnHook
\end{hscode}\resethooks
Note that unification might make no progress at all, and just return
\ensuremath{\Sigma } and the identity substitution.

\subsection{A suitable unification algorithm}

Higher order unification in the context of dependently typed languages
is still a poorly understood topic, and we do not have the space to
discuss it in depth here.  The basis for most of the unification
algorithms employed in such languages is \emph{pattern} unification, as
introduced by Miller \cite{miller1992}. However in modern languages such
as Agda or Coq there are factors complicating this process.  The main
inconvenience, as already mentioned in Section~\ref{introduction}, is
the fact that we cannot always know if a constraint is solvable, since
solving other constraints might enable us to solve the current one.  A
algorithm that takes care of handling constraints under these
assumptions is usually called \emph{dynamic}.

However, another issue when unifying terms in these type theories is
that definitional equality (specified by conversion rules such as the ones
presented in this paper) often includes \ensuremath{\eta }-laws for functions and
product types, and possibly other type-directed conversion rules.  For
this reason our constraints equate typed expressions.\footnote{We also
  have the constraints to equate the types too, but this is only for
  brevity: if that wasn't the case, all we would have to do is to
  convert each \ensuremath{\Gamma \vdash \Varid{t}\mathbin{:}\Conid{A}\mathrel{=}\Varid{u}\mathbin{:}\Conid{B}} into \ensuremath{\{\mskip1.5mu \Gamma \vdash \Conid{A}\mathbin{:}\mathsf{Set} \mathrel{=}\Conid{B}\mathbin{:}\mathsf{Set} ,\Gamma \vdash \Varid{t}\mathbin{:}\Conid{A}\mathrel{=}\Varid{u}\mathbin{:}\Conid{B}\mskip1.5mu\}}.} Note that the types in the
constraints are only needed so that we can abide by the \ensuremath{\eta } laws.

Ideally, to solve the heterogeneous constraints, a heterogeneous pattern
unification algorithm is needed, such as the one described in chapter 4 of Adam
Gundry's thesis \cite{gundry2013}.  However, in our prototype, we have
found that implementing such an algorithm is impractical for performance
reasons.\footnote{Specifically, the need to typecheck meta-variable
  instantiations and retry if they are ill-typed is particularly
  taxing.}  In practice, a homogeneous pattern unification algorithm
like the one employed in the Agda system works well enough in many of
the examples that we have analysed.  In this context, a heterogeneous
constraint \ensuremath{\Gamma \vdash \Varid{t}\mathbin{:}\Conid{A}\mathrel{=}\Varid{u}\mathbin{:}\Conid{B}} is converted to two homogeneous
constraints, \ensuremath{\Gamma \vdash \Conid{A}\mathrel{=}\Conid{B}\mathbin{:}\mathsf{Set} } and \ensuremath{\Gamma \vdash \Varid{t}\mathrel{=}\Varid{u}\mathbin{:}\Conid{A}}.  Some
bookkeeping will be needed to ensure that the constraint equating the
types is solved before attempting the one equating the terms, so that
\ensuremath{\Gamma \vdash \Varid{t}\mathrel{=}\Varid{u}\mathbin{:}\Conid{A}} is considered only when it is well-formed---when we
know that \ensuremath{\Conid{A}\equiv \Conid{B}}.

\section{Type checking through unification}
\label{algorithm}

As mentioned in Section~\ref{problem}, our algorithm will elaborate a
type-checking problem into a well-typed term and a set of unification
constraints. Given some type-checking problem \ensuremath{\Sigma ;\Gamma \vdash \Varid{t}\mathbin{:}\Conid{A}}, the
algorithm will elaborate it into a term \ensuremath{\Varid{u}} and constraints \ensuremath{\mathcal{C} }, along
with an extended signature \ensuremath{\Xi }---since the elaboration process will add
new meta-variables.

The algorithm is specified using rules of the form
\begin{hscode}\SaveRestoreHook
\column{B}{@{}>{\hspre}l<{\hspost}@{}}%
\column{3}{@{}>{\hspre}l<{\hspost}@{}}%
\column{E}{@{}>{\hspre}l<{\hspost}@{}}%
\>[3]{}\llbracket \Sigma ;\Gamma \vdash \Varid{t}\mathbin{:}\Conid{A}\rrbracket \leadsto \Xi ,\Varid{u},\mathcal{C} ,{}\<[E]%
\ColumnHook
\end{hscode}\resethooks
such that:\mytodo{I should name these properties (or at least the last
  one) here, and refer to them in \ref{some-properties}.}
\begin{itemize}
  \item \ensuremath{\Xi } is an extension of \ensuremath{\Sigma };
  \item \ensuremath{\Varid{u}} is well typed: \ensuremath{\Xi ;\Gamma \vdash \Varid{u}\mathbin{:}\Conid{A}};
  \item \ensuremath{\mathcal{C} } is a set of well formed unification constraints;
  \item if unification produces a signature \ensuremath{\Xi' } and substitution
    \ensuremath{\Xi' \vdash \theta \mathbin{:}\Xi } such that \ensuremath{\Xi' \Vdash \theta  \mathcal{C} }, we have that
    \ensuremath{\Xi' ;\theta  \Gamma \vdash \theta  \Varid{t}\equiv \theta  \Varid{u}\mathbin{:}\theta  \Conid{A}}.  In
    other words, solving all the constraints restores definitional
    equality between the original term \ensuremath{\Varid{t}} and the original term \ensuremath{\Varid{u}}.
\end{itemize}

The main idea is to infer what the type must look like based on the
term, and generate constraints that make sure that, if the constraints
can be solved, that will be the case.  For example, if faced with
problem
\begin{hscode}\SaveRestoreHook
\column{B}{@{}>{\hspre}l<{\hspost}@{}}%
\column{3}{@{}>{\hspre}l<{\hspost}@{}}%
\column{E}{@{}>{\hspre}l<{\hspost}@{}}%
\>[3]{}\Sigma ;\Gamma \vdash \mathsf{true} \mathbin{:}\Conid{A},{}\<[E]%
\ColumnHook
\end{hscode}\resethooks
we know that \ensuremath{\Conid{A}} must be \ensuremath{\mathsf{Bool} }.  However, \ensuremath{\Conid{A}} might be a type stuck on
a meta-variable, as discussed in Section~\ref{problem}.  At the same
time, we want the elaboration procedure to immediately return a
well-typed term.  The heterogeneous constraints let us do just that: we
will create a new meta-variable of type \ensuremath{\Conid{A}} in \ensuremath{\Gamma }, and use that as
the elaborated term.  Moreover we will return a constraint equating the
newly created meta-variable to \ensuremath{\mathsf{true} }:
\begin{hscode}\SaveRestoreHook
\column{B}{@{}>{\hspre}l<{\hspost}@{}}%
\column{3}{@{}>{\hspre}l<{\hspost}@{}}%
\column{5}{@{}>{\hspre}l<{\hspost}@{}}%
\column{E}{@{}>{\hspre}l<{\hspost}@{}}%
\>[3]{}\llbracket \Sigma ;\Gamma \vdash \mathsf{true} \mathbin{:}\Conid{A}\rrbracket \leadsto {}\<[E]%
\\
\>[3]{}\hsindent{2}{}\<[5]%
\>[5]{}(\Sigma ,\alpha \mathbin{:}\Gamma \to \Conid{A}),\;\alpha \;\Gamma ,\;\{\mskip1.5mu \Gamma \vdash \mathsf{true} \mathbin{:}\mathsf{Bool} \mathrel{=}\alpha \;\Gamma \mathbin{:}\Conid{A}\mskip1.5mu\}{}\<[E]%
\ColumnHook
\end{hscode}\resethooks
Note how we respect the contract of elaboration---the elaborated term is
well-typed, the constraint is well-formed---without making any
commitment on the shape of \ensuremath{\Conid{A}}.

For a more complicated example, consider the type-checking problem
\begin{hscode}\SaveRestoreHook
\column{B}{@{}>{\hspre}l<{\hspost}@{}}%
\column{3}{@{}>{\hspre}l<{\hspost}@{}}%
\column{E}{@{}>{\hspre}l<{\hspost}@{}}%
\>[3]{}\Sigma ;\Gamma \vdash \lambda \Varid{x}\to \Varid{t}\mathbin{:}\Conid{A}.{}\<[E]%
\ColumnHook
\end{hscode}\resethooks
We know that \ensuremath{\Conid{A}} needs to be a function type, but we do not know what
the domain and codomain types are yet.  To get around this problem we
will add new meta-variables to the signature acting as the domain and
codomain, then elaborate the body using those, and follow the same
technique illustrated above to return a well-typed term by adding a new
meta-variable:
\begin{hscode}\SaveRestoreHook
\column{B}{@{}>{\hspre}l<{\hspost}@{}}%
\column{3}{@{}>{\hspre}l<{\hspost}@{}}%
\column{E}{@{}>{\hspre}l<{\hspost}@{}}%
\>[3]{}\mbox{\onelinecomment  Add meta-variables for the domain (\ensuremath{\beta }) and codomain (\ensuremath{\gamma }):}{}\<[E]%
\\
\>[3]{}\Sigma_1 \mathbin{:=}\Sigma ,\beta \mathbin{:}\Gamma \to \mathsf{Set} ,\gamma \mathbin{:}\Gamma \to (\Varid{x}\mathbin{:}\beta )\to \mathsf{Set} {}\<[E]%
\\
\>[3]{}\mbox{\onelinecomment  Elaborate the body of the abstraction:}{}\<[E]%
\\
\>[3]{}\llbracket \Sigma_1 ;\Gamma ,\Varid{x}\mathbin{:}\beta \vdash \Varid{t}\mathbin{:}\gamma \;\Gamma \;\Varid{x}\rrbracket \leadsto \Sigma_2 ,\Varid{u},\mathcal{C} {}\<[E]%
\\
\>[3]{}\mbox{\onelinecomment  Add meta-variable that we will return as the elaborated term:}{}\<[E]%
\\
\>[3]{}\Sigma_3 \mathbin{:=}\Sigma_2 ,\alpha \mathbin{:}\Gamma \to \Conid{A}{}\<[E]%
\\
\>[3]{}\mbox{\onelinecomment  Return the appropriate constraint equating the abstracted elaborated body to the new}{}\<[E]%
\\
\>[3]{}\mbox{\onelinecomment  meta-variable, together with the constraints generated from elaborating the body:}{}\<[E]%
\\
\>[3]{}\llbracket \Sigma ;\Gamma \vdash \lambda \Varid{x}\to \Varid{t}\mathbin{:}\Conid{A}\rrbracket \leadsto \Sigma_4 ,{}\<[E]%
\\
\>[3]{}\alpha \;\Gamma ,{}\<[E]%
\\
\>[3]{}\{\mskip1.5mu \Gamma \vdash (\lambda \Varid{x}\to \Varid{u})\mathbin{:}(\Varid{x}\mathbin{:}\beta \;\Gamma )\to \gamma \;\Gamma \;\Varid{x}\mathrel{=}\alpha \;\Gamma \mathbin{:}\Conid{A}\mskip1.5mu\}\cup \mathcal{C} {}\<[E]%
\ColumnHook
\end{hscode}\resethooks
Note how the use of heterogeneous equality is crucial if we want to
avoid to ever commit to the types having a particular shape, while
having all the constraints to be well-formed immediately.

Every rule in full algorithm is going to follow the general pattern that
emerged above:
\begin{enumerate}[(a)]
\item Elaborate sub-terms, adding meta-variables to have types to work
  with;
\item Add a meta-variable serving as the elaborated term, say \ensuremath{\alpha };
\item Return a constraint equating \ensuremath{\alpha } to a term properly
  constructed using the elaborated subterms; together with the
  constraints returned by elaborating said subterms.
\end{enumerate}

For brevity we present the full algorithm using an abbreviated notation
that implicitely threads the signatures (\ensuremath{\Sigma }, \ensuremath{\Sigma_1 }, \ensuremath{\Sigma_2 }, and
\ensuremath{\Sigma_3 } in the example above) across rules.  We will also use the macro
\ensuremath{\textsc{Fresh}(\_ ,\_ )} to add new meta-variables in a context, such that \ensuremath{\alpha \mathbin{:=}\textsc{Fresh}(\Gamma ,\Conid{A})} is equivalent to \ensuremath{\Xi \mathbin{:=}\Sigma ,\alpha \mathbin{:}\Gamma \to \Conid{A}}, and
successive appearances of \ensuremath{\alpha } are automatically applied to \ensuremath{\Gamma },
and where \ensuremath{\Sigma } and \ensuremath{\Xi } are the old and new signature---which are, as
said, implicitly threaded.  Finally, we implicitly collect the
constraints generated when elaborating the subterms, and implicitly add
a meta-variable that stands for the elaborated term, together with an
appropriate constraint---steps (b) and (c) in the process described
above.\footnote{In our Haskell implementation, the elaboration is a
  state and writer monad over the signature and the constraint set,
  respectively.} For example, the rule to elaborate abstractions,
explained before, will be shown as
\[
\inference{
  \ensuremath{\beta \mathbin{:=}\textsc{Fresh}(\Gamma ,\mathsf{Set} )} & \ensuremath{\gamma \mathbin{:=}\textsc{Fresh}(\Gamma ,\Varid{x}\mathbin{:}\beta ,\mathsf{Set} )} \\
  \ensuremath{\llbracket \Gamma ,\Varid{x}\mathbin{:}\beta \vdash \Varid{t}\mathbin{:}\gamma \rrbracket \leadsto \Varid{u}}
}{
  \ensuremath{\llbracket \Gamma \vdash \lambda \Varid{x}\to \Varid{t}\mathbin{:}\Conid{A}\rrbracket \leadsto (\lambda \Varid{x}\to \Varid{u})\mathbin{:}(\Varid{x}\mathbin{:}\beta )\to \gamma }
}
\]

The complete algorithm is shown in figure \ref{elaboration}.  They are
remarkably similar to the typing rules, however instead of matching
directly on the type we expect we match through constraints.

The rules can easily be turned into an algorithm which pattern matches
on the term and decides how to proceed.\footnote{Our implementation can
  be found at
  \url{https://github.com/bitonic/tog/blob/master/src/Tog/Elaborate.hs}.}
Naturally the algorithm can still fail if it encounters an out of scope
meta-variable or variable, although in real systems scope checking is
usually performed beforehand.

\begin{figure}
  \[
  \inference{}{
    \ensuremath{\llbracket \Gamma \vdash \mathsf{Set} \mathbin{:}\Conid{A}\rrbracket \leadsto \mathsf{Set} \mathbin{:}\mathsf{Set} }
  }\hfill
  \inference{}{
    \ensuremath{\llbracket \Gamma \vdash \mathsf{Bool} \mathbin{:}\Conid{A}\rrbracket \leadsto \mathsf{Bool} \mathbin{:}\mathsf{Set} }
  }\hfill
  \inference{}{
    \ensuremath{\llbracket \Gamma \vdash \mathsf{true} \mathbin{:}\Conid{A}\rrbracket \leadsto \mathsf{true} \mathbin{:}\mathsf{Bool} }
  }
  \]
  \[
  \inference{}{
    \ensuremath{\llbracket \Gamma \vdash \mathsf{false} \mathbin{:}\Conid{A}\rrbracket \leadsto \mathsf{false} \mathbin{:}\mathsf{Bool} }
  }\hfill
  \inference{
    \ensuremath{\llbracket \Gamma \vdash \Conid{A}\mathbin{:}\mathsf{Set} \rrbracket \leadsto \Conid{A'}} & \ensuremath{\Gamma ,\Varid{x}\mathbin{:}\Conid{A'}\vdash \Conid{B}\mathbin{:}\mathsf{Set} \leadsto \Conid{B'}}
  }{
    \ensuremath{\llbracket \Gamma \vdash (\Varid{x}\mathbin{:}\Conid{A})\to \Conid{B}\mathbin{:}\Conid{S}\rrbracket \leadsto (\Varid{x}\mathbin{:}\Conid{A'})\to \Conid{B'}\mathbin{:}\mathsf{Set} }
  }
  \]
  \[
  \inference{
    \ensuremath{\beta \mathbin{:=}\textsc{Fresh}(\Gamma ,\mathsf{Set} )} & \ensuremath{\gamma \mathbin{:=}\textsc{Fresh}(\Gamma ,\Varid{x}\mathbin{:}\beta ,\mathsf{Set} )} \\
    \ensuremath{\llbracket \Gamma ,\Varid{x}\mathbin{:}\beta \vdash \Varid{t}\mathbin{:}\gamma \rrbracket \leadsto \Varid{u}}
  }{
    \ensuremath{\llbracket \Gamma \vdash \lambda \Varid{x}\to \Varid{t}\mathbin{:}\Conid{A}\rrbracket \leadsto (\lambda \Varid{x}\to \Varid{u})\mathbin{:}(\Varid{x}\mathbin{:}\beta )\to \gamma }
  }\hfill
  \inference{
    \ensuremath{\Varid{x}\mathbin{:}\Conid{A}\in \Gamma }
  }{
    \ensuremath{\llbracket \Gamma \vdash \Varid{x}\rrbracket \leadsto \Varid{x}\mathbin{:}\Conid{A}}
  }
  \]
  \[
  \inference{
    \ensuremath{\alpha \mathbin{:}\Conid{A}\in \Sigma }
  }{
    \ensuremath{\llbracket \Gamma \vdash \alpha \rrbracket \leadsto \alpha \mathbin{:}\Conid{A}}
  }\hfill
  \inference{
    \ensuremath{\beta \mathbin{:=}\textsc{Fresh}(\Gamma ,\mathsf{Set} )} & \ensuremath{\gamma \mathbin{:=}\textsc{Fresh}(\Gamma ,\Varid{x}\mathbin{:}\beta ,\mathsf{Set} )} \\
    \ensuremath{\llbracket \Gamma \vdash \Varid{n}\mathbin{:}(\Varid{x}\mathbin{:}\beta )\to \gamma \rrbracket \leadsto \Varid{t}} &
    \ensuremath{\llbracket \Gamma \vdash \Varid{u}\mathbin{:}\beta \rrbracket \leadsto \Varid{v}}
  }{
    \ensuremath{\llbracket \Gamma \vdash \Varid{n}\;\Varid{u}\mathbin{:}\Conid{A}\rrbracket \leadsto \Varid{t}\;\Varid{v}\mathbin{:}\gamma [\Varid{x}\mathbin{:=}\Varid{v}]}
  }
  \]
  \[
  \inference{
    \ensuremath{\llbracket \Gamma ,\Varid{x}\mathbin{:}\mathsf{Bool} \vdash \Conid{B}\mathbin{:}\mathsf{Set} \rrbracket \leadsto \Conid{B'}} & \ensuremath{\llbracket \Gamma \vdash \Varid{n}\mathbin{:}\mathsf{Bool} \rrbracket \leadsto \Varid{t}} \\
    \ensuremath{\llbracket \Gamma \vdash \Varid{u}\mathbin{:}\Conid{B'}[\Varid{x}\mathbin{:=}\mathsf{true} ]\rrbracket \leadsto \Varid{u'}} & \ensuremath{\llbracket \Gamma \vdash \Varid{v}\mathbin{:}\Conid{B'}[\Varid{x}\mathbin{:=}\mathsf{false} ]\rrbracket \leadsto \Varid{v'}}
  }{
    \ensuremath{\llbracket \Gamma \vdash \mathbf{if}\;\Varid{n}\;/\Varid{x}.\Conid{B}\;\mathbf{then}\;\Varid{u}\;\mathbf{else}\;\Varid{v}\mathbin{:}\Conid{A}\rrbracket \leadsto \mathbf{if}\;\Varid{t}\;/\Varid{x}.\Conid{B'}\;\mathbf{then}\;\Varid{u'}\;\mathbf{else}\;\Varid{v'}\mathbin{:}\Conid{B'}[\Varid{x}\mathbin{:=}\Varid{t}]}
  }
  \]
  \caption{\boxed{\ensuremath{\llbracket \Sigma ;\Gamma \vdash \Varid{t}\mathbin{:}\Conid{A}\rrbracket \leadsto \Xi ,\Varid{u},\mathcal{C} }} Elaboration}
  \label{elaboration}
\end{figure}

\subsection{Examples}

We will explore how the algorithm works by going through various common
situations.  The reader can experiment using the mentioned \texttt{tog}
tool, passing the \texttt{-d 'elaborate'} command-line flag to have it
to print out the generated constraints.  A wealth of examples are
present in the repository.  We will assume the usage of pattern
unification to solve constraints when examining the examples.

\subsubsection{A simple problem}

Let's take type-checking problem
\begin{hscode}\SaveRestoreHook
\column{B}{@{}>{\hspre}l<{\hspost}@{}}%
\column{3}{@{}>{\hspre}l<{\hspost}@{}}%
\column{E}{@{}>{\hspre}l<{\hspost}@{}}%
\>[3]{}\cdot ;\cdot \vdash \mathsf{true} \mathbin{:}\mathsf{Bool} .{}\<[E]%
\ColumnHook
\end{hscode}\resethooks
The algorithm will return the triple
\begin{hscode}\SaveRestoreHook
\column{B}{@{}>{\hspre}l<{\hspost}@{}}%
\column{3}{@{}>{\hspre}l<{\hspost}@{}}%
\column{E}{@{}>{\hspre}l<{\hspost}@{}}%
\>[3]{}\alpha \mathbin{:}\mathsf{Bool} ,\alpha ,\{\mskip1.5mu \cdot \vdash \mathsf{true} \mathbin{:}\mathsf{Bool} \mathrel{=}\alpha \mathbin{:}\mathsf{Bool} \mskip1.5mu\}{}\<[E]%
\ColumnHook
\end{hscode}\resethooks
The constraint is immediately solvable, yielding the substitution \ensuremath{\alpha \mathbin{:=}\mathsf{true} }, which will restore definitional equality between the
elaborated term and \ensuremath{\mathsf{true} }.

\subsubsection{An ill-typed problem}

Now for something that should fail. Given
\begin{hscode}\SaveRestoreHook
\column{B}{@{}>{\hspre}l<{\hspost}@{}}%
\column{3}{@{}>{\hspre}l<{\hspost}@{}}%
\column{E}{@{}>{\hspre}l<{\hspost}@{}}%
\>[3]{}\Varid{add}\mathbin{:}\mathsf{Nat} \to \mathsf{Nat} \to \mathsf{Nat} ,{}\<[E]%
\ColumnHook
\end{hscode}\resethooks
we want to solve
\begin{hscode}\SaveRestoreHook
\column{B}{@{}>{\hspre}l<{\hspost}@{}}%
\column{3}{@{}>{\hspre}l<{\hspost}@{}}%
\column{E}{@{}>{\hspre}l<{\hspost}@{}}%
\>[3]{}\cdot ;\Varid{x}\mathbin{:}\mathsf{Nat} \vdash \Varid{add}\;\Varid{x}\mathbin{:}\mathsf{Nat} .{}\<[E]%
\ColumnHook
\end{hscode}\resethooks
The algorithm will return the triple
\begin{hscode}\SaveRestoreHook
\column{B}{@{}>{\hspre}l<{\hspost}@{}}%
\column{3}{@{}>{\hspre}l<{\hspost}@{}}%
\column{7}{@{}>{\hspre}l<{\hspost}@{}}%
\column{13}{@{}>{\hspre}l<{\hspost}@{}}%
\column{16}{@{}>{\hspre}l<{\hspost}@{}}%
\column{E}{@{}>{\hspre}l<{\hspost}@{}}%
\>[3]{}\Sigma ,\zeta \;\Varid{x},\mathcal{C} \;\mathbf{where}{}\<[E]%
\\
\>[3]{}\hsindent{4}{}\<[7]%
\>[7]{}\Sigma \mathrel{=}{}\<[13]%
\>[13]{}\alpha \mathbin{:}\mathsf{Nat} \to \mathsf{Nat} ,\;\beta \mathbin{:}(\Varid{x}\mathbin{:}\mathsf{Nat} )\to \mathsf{Nat} \to \mathsf{Set} ,\;\gamma \mathbin{:}(\Varid{x}\mathbin{:}\mathsf{Nat} )\to \alpha \;\Varid{x},{}\<[E]%
\\
\>[13]{}\delta \mathbin{:}(\Varid{x}\mathbin{:}\mathsf{Nat} )\to (\Varid{y}\mathbin{:}\alpha \;\Varid{x})\to \beta \;\Varid{x}\;\Varid{y},\;\zeta \mathbin{:}\mathsf{Nat} \to \mathsf{Nat} {}\<[E]%
\\[\blanklineskip]%
\>[3]{}\hsindent{4}{}\<[7]%
\>[7]{}\mathcal{C} \mathrel{=}\{\mskip1.5mu {}\<[16]%
\>[16]{}\Varid{x}\mathbin{:}\mathsf{Nat} \vdash \delta \;\Varid{x}\;(\gamma \;\Varid{x})\mathbin{:}\beta \;\Varid{x}\;(\gamma \;\Varid{x})\mathrel{=}\zeta \;\Varid{x}\mathbin{:}\mathsf{Nat} ,{}\<[E]%
\\
\>[16]{}\Varid{x}\mathbin{:}\mathsf{Nat} \vdash \Varid{add}\mathbin{:}\mathsf{Nat} \to \mathsf{Nat} \to \mathsf{Nat} \mathrel{=}\delta \;\Varid{x}\mathbin{:}(\Varid{y}\mathbin{:}\alpha \;\Varid{x})\to \beta \;\Varid{x}\;\Varid{y},{}\<[E]%
\\
\>[16]{}\Varid{x}\mathbin{:}\mathsf{Nat} \vdash \Varid{x}\mathbin{:}\mathsf{Nat} \mathrel{=}\gamma \;\Varid{x}\mathbin{:}\alpha \;\Varid{x}\mskip1.5mu\}{}\<[E]%
\ColumnHook
\end{hscode}\resethooks
While looking scary at first, the meaning of the meta-variables and
constraints is easy to interpret, keeping in mind that we generate one
constraint per subterm (including the top level term).

At the top level, we elaborate the two subterms \ensuremath{\Varid{add}} and \ensuremath{\Varid{x}}.  We know
that \ensuremath{\Varid{add}} must be a function type, since it is applied to something;
and that the type of \ensuremath{\Varid{x}} must match the type of the domain of said
function type. Thus, two meta-variables are created to represent the
domain and codomain---\ensuremath{\alpha } and \ensuremath{\beta }.  Then, \ensuremath{\Varid{add}} and \ensuremath{\Varid{x}} are
elaborated with said types, which in turn requires the addition of
\ensuremath{\gamma } and \ensuremath{\delta }, serving as elaborated terms.  These are the
ingredients for the second and third constraint. Finally, the elaborated
\ensuremath{\delta } (representing \ensuremath{\Varid{add}}) is applied to \ensuremath{\gamma } (representing \ensuremath{\Varid{x}}),
and equated to a new meta-variable \ensuremath{\zeta }, which is the result of the
top-level elaboration.

The constraints reflect the fact that the term is ill-typed: \ensuremath{\beta } is
equated both to \ensuremath{\mathsf{Nat} } (in the first constraint), and to \ensuremath{\mathsf{Nat} \to \mathsf{Nat} }, in
the second constraint.  Thus, unification will fail.

\subsubsection{An unsolvable problem}

Finally, let's go back to an example which cannot immediately be solved
in its entirety:
\begin{hscode}\SaveRestoreHook
\column{B}{@{}>{\hspre}l<{\hspost}@{}}%
\column{3}{@{}>{\hspre}l<{\hspost}@{}}%
\column{E}{@{}>{\hspre}l<{\hspost}@{}}%
\>[3]{}\alpha \mathbin{:}\mathsf{Nat} ;\cdot \vdash (\mathsf{true} ,\mathrm{0})\mathbin{:}\Conid{BoolOrNat}\;\alpha \times \mathsf{Nat} .{}\<[E]%
\ColumnHook
\end{hscode}\resethooks
The desired outcome for this problem is to type check the second element
of the pair, \ensuremath{\mathrm{0}}; but ``suspend'' the type checking of the first element
by replacing \ensuremath{\mathsf{refl} } with a meta-variable.  Running the type-checking
problem through the elaboration procedure yields
\begin{hscode}\SaveRestoreHook
\column{B}{@{}>{\hspre}l<{\hspost}@{}}%
\column{3}{@{}>{\hspre}l<{\hspost}@{}}%
\column{7}{@{}>{\hspre}l<{\hspost}@{}}%
\column{13}{@{}>{\hspre}l<{\hspost}@{}}%
\column{14}{@{}>{\hspre}c<{\hspost}@{}}%
\column{14E}{@{}l@{}}%
\column{17}{@{}>{\hspre}l<{\hspost}@{}}%
\column{E}{@{}>{\hspre}l<{\hspost}@{}}%
\>[3]{}\Sigma ,\zeta \;\Varid{x},\mathcal{C} \;\mathbf{where}{}\<[E]%
\\
\>[3]{}\hsindent{4}{}\<[7]%
\>[7]{}\Sigma \mathrel{=}{}\<[13]%
\>[13]{}\beta \mathbin{:}\mathsf{Set} ,\;\gamma \mathbin{:}\mathsf{Set} ,\;\delta \mathbin{:}\beta ,\;\zeta \mathbin{:}\gamma ,\;\iota \mathbin{:}\Conid{BoolOrNat}\;\alpha \times \mathsf{Nat} {}\<[E]%
\\[\blanklineskip]%
\>[3]{}\hsindent{4}{}\<[7]%
\>[7]{}\mathcal{C} \mathrel{=}{}\<[14]%
\>[14]{}\{\mskip1.5mu {}\<[14E]%
\>[17]{}\cdot \vdash (\delta ,\zeta )\mathbin{:}\beta \times \gamma \mathrel{=}\iota \mathbin{:}\Conid{BoolOrNat}\;\alpha \times \mathsf{Nat} ,{}\<[E]%
\\
\>[17]{}\cdot \vdash \mathrm{0}\mathbin{:}\mathsf{Nat} \mathrel{=}\zeta \mathbin{:}\gamma ,{}\<[E]%
\\
\>[17]{}\cdot \vdash \mathsf{true} \mathbin{:}\mathsf{Bool} \mathrel{=}\delta \mathbin{:}\beta \mskip1.5mu\}{}\<[E]%
\ColumnHook
\end{hscode}\resethooks
The second and third constraints regard the two sub-problems
individually, and are solvable yielding the substitution
\begin{hscode}\SaveRestoreHook
\column{B}{@{}>{\hspre}l<{\hspost}@{}}%
\column{3}{@{}>{\hspre}l<{\hspost}@{}}%
\column{E}{@{}>{\hspre}l<{\hspost}@{}}%
\>[3]{}\beta \mathbin{:=}\mathsf{Bool} ,\delta \mathbin{:=}\mathsf{true} ,\gamma \mathbin{:=}\mathsf{Nat} ,\zeta \mathbin{:=}\mathrm{0},{}\<[E]%
\ColumnHook
\end{hscode}\resethooks
and leaving us with
\begin{hscode}\SaveRestoreHook
\column{B}{@{}>{\hspre}l<{\hspost}@{}}%
\column{3}{@{}>{\hspre}l<{\hspost}@{}}%
\column{E}{@{}>{\hspre}l<{\hspost}@{}}%
\>[3]{}\cdot \vdash (\mathsf{true} ,\mathrm{0})\mathbin{:}\mathsf{Bool} \times \mathsf{Nat} \mathrel{=}\iota \mathbin{:}\Conid{BoolOrNat}\;\alpha \times \mathsf{Nat} .{}\<[E]%
\ColumnHook
\end{hscode}\resethooks
At this point the unifier will be able to substitute \ensuremath{\iota } with a pair
\begin{hscode}\SaveRestoreHook
\column{B}{@{}>{\hspre}l<{\hspost}@{}}%
\column{3}{@{}>{\hspre}l<{\hspost}@{}}%
\column{E}{@{}>{\hspre}l<{\hspost}@{}}%
\>[3]{}\cdot \vdash (\mathsf{true} ,\mathrm{0})\mathbin{:}\mathsf{Bool} \times \mathsf{Nat} \mathrel{=}(\xi ,\mathrm{0})\mathbin{:}\Conid{BoolOrNat}\;\alpha \times \mathsf{Nat} .{}\<[E]%
\ColumnHook
\end{hscode}\resethooks
Now solving the constraint amounts to solve
\begin{hscode}\SaveRestoreHook
\column{B}{@{}>{\hspre}l<{\hspost}@{}}%
\column{3}{@{}>{\hspre}l<{\hspost}@{}}%
\column{E}{@{}>{\hspre}l<{\hspost}@{}}%
\>[3]{}\cdot \vdash \mathsf{true} \mathbin{:}\mathsf{Bool} \mathrel{=}\xi \mathbin{:}\Conid{BoolOrNat}\;\alpha ,{}\<[E]%
\ColumnHook
\end{hscode}\resethooks
but the unifier is not going to be able to make progress, since
substituting \ensuremath{\xi } for \ensuremath{\mathsf{true} } will leave the constraint ill-formed.
Thus, the final result of the elaboration and unification will be \ensuremath{(\xi ,\mathrm{0})}, which is all we could hope for.

\subsection{Some properties}
\label{some-properties}

\begin{lemma}[Well-formedness and functionality of elaboration]
  Let  \ensuremath{\vdash \Sigma } and \ensuremath{\Sigma \vdash \Gamma }.
  \begin{enumerate}
  \item
  There uniquely exist
  a signature \ensuremath{\Xi }, constraints \ensuremath{\mathcal{C} }, a term \ensuremath{\Varid{t'}}, and a type \ensuremath{\Conid{A}}  such that
  \begin{hscode}\SaveRestoreHook
\column{B}{@{}>{\hspre}l<{\hspost}@{}}%
\column{5}{@{}>{\hspre}l<{\hspost}@{}}%
\column{E}{@{}>{\hspre}l<{\hspost}@{}}%
\>[5]{}\llbracket \Sigma ;\Gamma \vdash \Varid{t}\rrbracket \leadsto \Xi \ |\  \mathcal{C} \vdash \Varid{t'}\Rightarrow \Conid{A}{}\<[E]%
\ColumnHook
\end{hscode}\resethooks
  and
  \ensuremath{\Xi ;\Gamma \vdash \Conid{A}\mathbin{:}\mathsf{Set} }.
  \item
  If \ensuremath{\Sigma ;\Gamma \vdash \Conid{A}\mathbin{:}\mathsf{Set} }, then there uniquely exist
  a signature \ensuremath{\Xi }, constraints \ensuremath{\mathcal{C} }, and a term \ensuremath{\Varid{t'}} such that
  \begin{hscode}\SaveRestoreHook
\column{B}{@{}>{\hspre}l<{\hspost}@{}}%
\column{5}{@{}>{\hspre}l<{\hspost}@{}}%
\column{E}{@{}>{\hspre}l<{\hspost}@{}}%
\>[5]{}\llbracket \Sigma ;\Gamma \vdash \Varid{t}\Leftarrow \Conid{A}\rrbracket \leadsto \Xi \ |\  \mathcal{C} \vdash \Varid{t'}.{}\<[E]%
\ColumnHook
\end{hscode}\resethooks
  \end{enumerate}
In both cases, \ensuremath{\Xi } is necessarily an extension of \ensuremath{\Sigma } and
all the outputs are well-formed, meaning
  \ensuremath{\vdash \Xi } and
  \ensuremath{\Xi ;\Gamma \vdash \Varid{t'}\mathbin{:}\Conid{A}} and
  \ensuremath{\Xi ;\Gamma \vdash \mathcal{C} }.
\end{lemma}
\begin{proof}
  By induction on $t$.
\end{proof}

\begin{lemma}[Soundness of elaboration]
\hfill\\
  If one of
  \begin{hscode}\SaveRestoreHook
\column{B}{@{}>{\hspre}l<{\hspost}@{}}%
\column{5}{@{}>{\hspre}l<{\hspost}@{}}%
\column{E}{@{}>{\hspre}l<{\hspost}@{}}%
\>[5]{}\llbracket \Sigma ;\Gamma \vdash \Varid{t}\rrbracket \leadsto \Xi \ |\  \mathcal{C} \vdash \Varid{t'}\Rightarrow \Conid{A}{}\<[E]%
\\
\>[5]{}\llbracket \Sigma ;\Gamma \vdash \Varid{t}\Leftarrow \Conid{A}\rrbracket \leadsto \Xi \ |\  \mathcal{C} \vdash \Varid{t'}{}\<[E]%
\ColumnHook
\end{hscode}\resethooks
  and there is some well-typed meta substitution
  \ensuremath{\Sigma' \vdash \theta \mathbin{:}\Xi } that solves the constraints,
  i.e., \ensuremath{\Sigma' \Vdash \theta  \mathcal{C} }, then
  \begin{hscode}\SaveRestoreHook
\column{B}{@{}>{\hspre}l<{\hspost}@{}}%
\column{5}{@{}>{\hspre}l<{\hspost}@{}}%
\column{E}{@{}>{\hspre}l<{\hspost}@{}}%
\>[5]{}\Sigma' ,\theta  \Gamma \vdash \theta  \Varid{t}\equiv \theta  \Varid{t'}\mathbin{:}\theta  \Conid{A}.{}\<[E]%
\ColumnHook
\end{hscode}\resethooks
\end{lemma}
\begin{proof}
  By induction on term \ensuremath{\Varid{t}}.
\end{proof}

\begin{lemma}[Strong soundness of elaboration] \hfill
  \begin{enumerate}
  \item  If
    \begin{hscode}\SaveRestoreHook
\column{B}{@{}>{\hspre}l<{\hspost}@{}}%
\column{7}{@{}>{\hspre}l<{\hspost}@{}}%
\column{E}{@{}>{\hspre}l<{\hspost}@{}}%
\>[7]{}\llbracket \Sigma ;\Gamma \vdash \Varid{t}\rrbracket \leadsto \Xi \ |\  \mathcal{C} \vdash \Varid{t'}\Rightarrow \Conid{A}{}\<[E]%
\ColumnHook
\end{hscode}\resethooks
    and there is some closing untyped substitution $\sigma$ such that
    $\dom(\sigma) \supseteq \dom(\Xi)$ and
    $\der \sigma : \Sigma$ and
    $\sigma\Gamma \der \sigma A : \Set$ and
    $\sigma$ solves the untyped constraints, then
    \begin{hscode}\SaveRestoreHook
\column{B}{@{}>{\hspre}l<{\hspost}@{}}%
\column{5}{@{}>{\hspre}l<{\hspost}@{}}%
\column{E}{@{}>{\hspre}l<{\hspost}@{}}%
\>[5]{}\sigma  \Gamma \vdash \sigma  \Varid{t}\equiv \sigma  \Varid{t'}\mathbin{:}\sigma  \Conid{A}.{}\<[E]%
\ColumnHook
\end{hscode}\resethooks
  \item
  If
  \begin{hscode}\SaveRestoreHook
\column{B}{@{}>{\hspre}l<{\hspost}@{}}%
\column{5}{@{}>{\hspre}l<{\hspost}@{}}%
\column{49}{@{}>{\hspre}l<{\hspost}@{}}%
\column{E}{@{}>{\hspre}l<{\hspost}@{}}%
\>[5]{}\llbracket \Sigma ;\Gamma \vdash \Varid{t}\Leftarrow \Conid{A}\rrbracket \leadsto \Xi \ |\  \mathcal{C} \mathbin{|-}{}\<[49]%
\>[49]{}\Varid{t'}{}\<[E]%
\ColumnHook
\end{hscode}\resethooks
  and there is some closing untyped substitution \ensuremath{\sigma \mathbin{:}\Xi } that
  is well-typed for \ensuremath{\Sigma }, i.e, \ensuremath{\vdash \sigma \mathbin{:}\Sigma } and solves the
  untyped constraints,
  i.e., \ensuremath{\Vdash \sigma  \mathcal{C} }, then
  \begin{hscode}\SaveRestoreHook
\column{B}{@{}>{\hspre}l<{\hspost}@{}}%
\column{5}{@{}>{\hspre}l<{\hspost}@{}}%
\column{E}{@{}>{\hspre}l<{\hspost}@{}}%
\>[5]{}\sigma  \Gamma \vdash \sigma  \Varid{t}\equiv \sigma  \Varid{t'}\mathbin{:}\sigma  \Conid{A}.{}\<[E]%
\ColumnHook
\end{hscode}\resethooks
  \end{enumerate}
\end{lemma}
\begin{proof}
  By induction on term \ensuremath{\Varid{t}}.
  \begin{caselist}

  \nextcase Abstraction $\lambda x \to t$
  \[
  \inference{\den{(\Sigma, \alpha : (\Gamma \to \Set)); (\Gamma, x : (\alpha\,\Gamma)) \der t}
      \leadsto
      \Xi \mid \C \der t' \jinf B
    }{\den{\Sigma; \Gamma \der (\lambda x \to t)}
      \leadsto
      \Xi \mid \C \der (\lambda x \to t') \jinf ((x : \alpha\,\Gamma) \to B)
    }
  \]
  By assumption $\der \sigma : \Sigma$ and $\sigma\Gamma \der
  \sigma((x : \alpha\,\Gamma) \to B) : \Set$, which by inversion gives
  us $\sigma\Gamma \der (\sigma\alpha)\,\Gamma : \Set$ and
  $\sigma(\Gamma, x : \alpha\,\Gamma) \der \sigma B : \Set$.
  Thus $\der \sigma\alpha : (\sigma\Gamma \to \Set)$ and
  $\der \sigma : (\Sigma, \alpha : (\Gamma \to \Set))$.
  By induction hypothesis,
  $\sigma\Gamma, x : \sigma(\alpha\,\Gamma) \der \sigma t = \sigma t' :
  \sigma B$, thus by the $\xi$ rule for equality,
  $\sigma\Gamma \der \sigma(\lambda x \to t) = \sigma(\lambda x \to
  t') : \sigma((x \of \alpha\,\Gamma) \to B)$.

  \nextcase Application $t\,u$
  \[
  \inference{
    \den{\Sigma;\Gamma \der u} \leadsto \Sigma_1\mid\C_1\der u' \jinf A
    \\
    \den{\Sigma_1, \beta : ((\Gamma,x \of A) \to \Set);\Gamma, x \of A
         \der t \jchk ((x \of A) \to \beta\,\Gamma\,x)}
    \leadsto \Sigma_2\mid\C_2\der t'
    }{
    \den{\Sigma;\Gamma \der t\,u} \leadsto
    \Sigma_2 \mid \C_1 \cup \C_2 \der t'\,u' \jinf \beta\, \Gamma\, u'
    }
  \]
  By the first induction hypothesis,
  $\sigma\Gamma \der \sigma u = \sigma u' : \sigma A$.
  By the second induction hypothesis
  \end{caselist}
\end{proof}




\begin{lemma}[Completeness of elaboration]
  Let us assume a term \ensuremath{\Varid{t}} and a wellformed signature \ensuremath{\Sigma } and \ensuremath{\Sigma \vdash \Gamma } and a second signature \ensuremath{\Xi } and a substitution \ensuremath{\Xi \vdash \theta \mathbin{:}\Sigma } such that
  \ensuremath{\Xi ;\theta  \Gamma \vdash \theta  \Varid{t}\mathbin{:}\Conid{C}}.
  \begin{enumerate}

\item Then
  \begin{hscode}\SaveRestoreHook
\column{B}{@{}>{\hspre}l<{\hspost}@{}}%
\column{5}{@{}>{\hspre}l<{\hspost}@{}}%
\column{E}{@{}>{\hspre}l<{\hspost}@{}}%
\>[5]{}\llbracket \Sigma ;\Gamma \vdash \Varid{t}\rrbracket \leadsto (\Sigma ,\Sigma' )\ |\  \mathcal{C} \vdash \Varid{t'}\Rightarrow \Conid{B}{}\<[E]%
\ColumnHook
\end{hscode}\resethooks
  and there exists a substitution \ensuremath{\theta' } such that
  \begin{hscode}\SaveRestoreHook
\column{B}{@{}>{\hspre}l<{\hspost}@{}}%
\column{5}{@{}>{\hspre}l<{\hspost}@{}}%
\column{E}{@{}>{\hspre}l<{\hspost}@{}}%
\>[5]{}\Xi \vdash (\theta ,\theta' )\mathbin{:}(\Sigma ,\Sigma' ){}\<[E]%
\\
\>[5]{}\Xi \Vdash (\theta ,\theta' ) \mathcal{C} {}\<[E]%
\\
\>[5]{}\Xi ;\theta  \Gamma \vdash \Conid{C}\mathrel{=}(\theta ,\theta' ) \Conid{B}\mathbin{:}\mathsf{Set} {}\<[E]%
\\
\>[5]{}\Xi ;\theta  \Gamma \vdash \theta  \Varid{t}\mathrel{=}(\theta ,\theta' ) \Varid{t'}\mathbin{:}\Conid{C}.{}\<[E]%
\ColumnHook
\end{hscode}\resethooks
\item Further, if \ensuremath{\Sigma ;\Gamma \vdash \Conid{A}\mathbin{:}\mathsf{Set} } and
  \ensuremath{\Xi ;\theta  \Gamma \vdash \theta  \Conid{A}\mathrel{=}\Conid{C}\mathbin{:}\mathsf{Set} }, then
  \begin{hscode}\SaveRestoreHook
\column{B}{@{}>{\hspre}l<{\hspost}@{}}%
\column{5}{@{}>{\hspre}l<{\hspost}@{}}%
\column{E}{@{}>{\hspre}l<{\hspost}@{}}%
\>[5]{}\llbracket \Sigma ;\Gamma \vdash \Varid{t}\Leftarrow \Conid{A}\rrbracket \leadsto (\Sigma ,\Sigma' )\ |\  \mathcal{C} \vdash \Varid{t'}{}\<[E]%
\ColumnHook
\end{hscode}\resethooks
  and there exists a substitution \ensuremath{\theta' } such that
  \begin{hscode}\SaveRestoreHook
\column{B}{@{}>{\hspre}l<{\hspost}@{}}%
\column{5}{@{}>{\hspre}l<{\hspost}@{}}%
\column{E}{@{}>{\hspre}l<{\hspost}@{}}%
\>[5]{}\Xi \vdash (\theta ,\theta' )\mathbin{:}(\Sigma ,\Sigma' ){}\<[E]%
\\
\>[5]{}\Xi \Vdash (\theta ,\theta' ) \mathcal{C} {}\<[E]%
\\
\>[5]{}\Xi ;\theta  \Gamma \vdash \theta  \Varid{t}\mathrel{=}(\theta ,\theta' ) \Varid{t'}\mathbin{:}\Conid{C}.{}\<[E]%
\ColumnHook
\end{hscode}\resethooks
  \end{enumerate}
\end{lemma}
  In other words, if it's possible to instantiate some meta-variables to
  make \ensuremath{\Varid{t}} well-typed, then all the constraints generated by the
  elaboration procedure are solvable.
\begin{proof}
  By induction on \ensuremath{\Varid{t}}.
\end{proof}

\begin{corollary}[Simple inference]
  If
  \ensuremath{\Sigma ;\Gamma \vdash \Varid{t}\mathbin{:}\Conid{A}}
  and
  \ensuremath{\llbracket \Sigma ;\Gamma \vdash \Varid{t}\rrbracket \leadsto (\Sigma ,\Sigma' )\ |\  \mathcal{C} \vdash \Varid{t'}\Rightarrow \Conid{A'}}%
  ,
  then all the constraints in \ensuremath{\mathcal{C} } are solvable by some substitution
  \ensuremath{\Sigma \vdash \theta' \mathbin{:}\Sigma' } and \ensuremath{\Sigma ;\Gamma \vdash \Conid{A}\mathrel{=}\theta'  \Conid{A'}\mathbin{:}\mathsf{Set} }.
\end{corollary}
\begin{proof}
  From Completeness with \ensuremath{\Xi \mathrel{=}\Sigma } and identity substitution \ensuremath{\theta }.
\end{proof}

\begin{corollary}[Simple checking]
  If
  \ensuremath{\Sigma ;\Gamma \vdash \Varid{t}\mathbin{:}\Conid{A}}
  and
  \ensuremath{\llbracket \Sigma ;\Gamma \vdash \Varid{t}\Leftarrow \Conid{A}\rrbracket \leadsto (\Sigma ,\Sigma' )\ |\  \mathcal{C} \vdash \Varid{t'}}%
  ,
  then all the constraints in \ensuremath{\mathcal{C} } are solvable.
\end{corollary}
\begin{proof}
  From Completeness with \ensuremath{\Xi \mathrel{=}\Sigma } and identity substitution \ensuremath{\theta }.
\end{proof}

\subsubsection{Effectiveness}

While the properties above guarantee establish some basic results about
the algorithm, they are all also satisfied by the very useless
elaboration algorithm
\begin{hscode}\SaveRestoreHook
\column{B}{@{}>{\hspre}l<{\hspost}@{}}%
\column{3}{@{}>{\hspre}l<{\hspost}@{}}%
\column{E}{@{}>{\hspre}l<{\hspost}@{}}%
\>[3]{}\llbracket \Sigma ;\Gamma \vdash \Varid{t}\mathbin{:}\Conid{A}\rrbracket \leadsto (\Sigma ,\alpha \mathbin{:}\Gamma \to \Conid{A},\beta \mathbin{:}\Gamma \to \mathsf{Set} ),\alpha \;\Gamma ,\{\mskip1.5mu \Gamma \vdash \alpha \;\Gamma \mathbin{:}\Conid{A}\mathrel{=}\Varid{t}\mathbin{:}\beta \;\Gamma \mskip1.5mu\}.{}\<[E]%
\ColumnHook
\end{hscode}\resethooks
However, the intent of the developed algorithm is to be as effective,
when used with pattern unification, as the current techniques to
type check dependent types with meta-variables.

To achieve this, it has been designed so that the generated constraints
fall in the pattern fragment when existing type checkers would be able
to make progress, and by testing the algorithm ``in the wild'' with
existing Agda programs we have found it effective in practice.  For
instance, if we did not keep our terms in \ensuremath{\beta }-normal form, we could
not elaborate applications into constraints falling into the pattern
fragment, due to the fact that we cannot reliably infer the type of the
function.

However, future work should involve a formal characterization of
completeness in the context of type checkers for dependent types with
meta-variables, and a proof that our algorithm is indeed complete.

\section{Additional remarks}
\label{remarks}

\subsection{Additional features}
\label{more-features}

In this section we will briefly sketch how to fit popular features into
the framework described.  The general idea is to do everything which
doesn't require normalization in the elaboration procedure, and the rest
into the unifier.

\subsubsection{Implicit arguments}

We have already implemented a restricted form of implicit arguments,
which we call ``type schemes'' in line with ML terminology.  Type
schemes allow the user to define a number of implicit arguments in
top-level definitions, for example the already mentioned
\begin{hscode}\SaveRestoreHook
\column{B}{@{}>{\hspre}l<{\hspost}@{}}%
\column{3}{@{}>{\hspre}l<{\hspost}@{}}%
\column{E}{@{}>{\hspre}l<{\hspost}@{}}%
\>[3]{}\Varid{head}\mathbin{:}\{\mskip1.5mu \Conid{A}\mathbin{:}\mathsf{Set} \mskip1.5mu\}\to \{\mskip1.5mu \Varid{n}\mathbin{:}\mathsf{Nat} \mskip1.5mu\}\to \mathsf{Vec} \;\Conid{A}\;(\mathrm{1}\mathbin{+}\Varid{n})\to \Conid{A}.{}\<[E]%
\ColumnHook
\end{hscode}\resethooks
Under this scheme, every occurrence of \ensuremath{\Varid{head}} is statically replaced
with \ensuremath{\Varid{head}\;\anonymous \;\anonymous }, before even elaborating---we do it while scope
checking. However this is obviously limited, since we can only have
implicit arguments before any explicit ones and only for top-level
definitions.

We want to enable a more liberal placement of implicit arguments.  This
is achieved in Agda by allowing implicit arguments in all types, and in
any many positions.  The details of the implementation are still in flux
in Agda itself, but the core idea is to type-check function application
and abstractions bidirectionally, by first looking at the type and
inserting implicit arguments if needed.  So elaborating \ensuremath{\Varid{t}\;\Varid{u}} where \ensuremath{\Varid{t}\mathbin{:}\{\mskip1.5mu \Varid{x}\mathbin{:}\Conid{A}\mskip1.5mu\}\to (\Varid{y}\mathbin{:}\Conid{B})\to \Conid{C}} will result in \ensuremath{\Varid{t}\;\anonymous \;\Varid{u}}.  Similarly, elaborating
\ensuremath{\Varid{t}\mathbin{:}\{\mskip1.5mu \Varid{x}\mathbin{:}\Conid{A}\mskip1.5mu\}\to \Conid{B}} will result in \ensuremath{\lambda \{\mskip1.5mu \Varid{x}\mskip1.5mu\}\to \Varid{t}}, where \ensuremath{\{\mskip1.5mu \Varid{x}\mskip1.5mu\}} binds an
implicit arguments.

We are exploring two ways to integrate this kind of mechanism in our
framework:
\begin{itemize}
\item Mimic the bidirectional type-checking performed in Agda and
  similar systems closely, by adding a new kind of constraint for
  function application and abstraction which waits on the type to have a
  \emph{rigid} head, that is to say a type not blocked on a
  meta-variable.
\item Alternatively, force all implicit arguments to appear before an
  explicit one (with the exception of type schemes), and always include
  an implicit arguments in \ensuremath{\to }-types.\footnote{Note that the
    restriction of having implicit arguments always paired with explicit
    ones is being planned in Agda too, to make the automatic insertion
    of implicit arguments more predictable.}  Multiple implicit
  arguments will then be handled by one implicit argument carrying a
  tuple.
\end{itemize}

We are currently implementing the latter proposal, since it is simpler
to describe and to implement, although only using it will tell if it is
practical.\footnote{A more thorough proposal is available at
  \url{https://github.com/bitonic/tog/wiki/Implicit-Arguments}.}

\subsubsection{Type classes}

Type classes, as employed by Haskell, were introduced to handle
overloaded operators in programming languages \cite{wadler1989}.  Other
similar structures include canonical structures in Coq.  In short, type
classes let the user specify a collection of methods that can be
implemented for a specific type.  For example, the type class
identifying monoid is defined as
\begin{hscode}\SaveRestoreHook
\column{B}{@{}>{\hspre}l<{\hspost}@{}}%
\column{3}{@{}>{\hspre}l<{\hspost}@{}}%
\column{5}{@{}>{\hspre}l<{\hspost}@{}}%
\column{14}{@{}>{\hspre}l<{\hspost}@{}}%
\column{22}{@{}>{\hspre}l<{\hspost}@{}}%
\column{E}{@{}>{\hspre}l<{\hspost}@{}}%
\>[3]{}\mathbf{class}\;\Conid{Monoid}\;\Varid{a}\;\mathbf{where}{}\<[E]%
\\
\>[3]{}\hsindent{2}{}\<[5]%
\>[5]{}\Varid{mempty}\mathbin{:}\Varid{a}{}\<[E]%
\\
\>[3]{}\hsindent{2}{}\<[5]%
\>[5]{}\Varid{mappend}\mathbin{:}\Varid{a}\to \Varid{a}\to \Varid{a}{}\<[E]%
\\[\blanklineskip]%
\>[3]{}\mbox{\onelinecomment  Monoid instance for list}{}\<[E]%
\\
\>[3]{}\mathbf{instance}\;\Conid{Monoid}\;[\mskip1.5mu \Varid{a}\mskip1.5mu]\;\mathbf{where}{}\<[E]%
\\
\>[3]{}\hsindent{2}{}\<[5]%
\>[5]{}\Varid{mempty}{}\<[14]%
\>[14]{}\mathrel{=}[\mskip1.5mu \mskip1.5mu]{}\<[22]%
\>[22]{}\mbox{\onelinecomment  Empty list}{}\<[E]%
\\
\>[3]{}\hsindent{2}{}\<[5]%
\>[5]{}\Varid{mappend}{}\<[14]%
\>[14]{}\mathrel{=}(\plus ){}\<[22]%
\>[22]{}\mbox{\onelinecomment  Concatenation}{}\<[E]%
\\[\blanklineskip]%
\>[3]{}\mbox{\onelinecomment  Monoid instance for pairs of monoids}{}\<[E]%
\\
\>[3]{}\mathbf{instance}\;(\Conid{Monoid}\;\Varid{a},\Conid{Monoid}\;\Varid{b})\Rightarrow \Conid{Monoid}\;(\Varid{a},\Varid{b})\;\mathbf{where}{}\<[E]%
\\
\>[3]{}\hsindent{2}{}\<[5]%
\>[5]{}\Varid{mempty}\mathrel{=}(\Varid{mempty},\Varid{mempty}){}\<[E]%
\\
\>[3]{}\hsindent{2}{}\<[5]%
\>[5]{}\Varid{mappend}\;(x_1 ,y_1 )\;(x_2 ,y_2 )\mathrel{=}(x_1 \mathbin{`\Varid{mappend}`}x_2 ,y_1 \mathbin{`\Varid{mappend}`}y_2 ){}\<[E]%
\ColumnHook
\end{hscode}\resethooks

Type classes are a form of type-directed name resolution, and thus
cannot be resolved at elaboration time---we might need to instantiate
meta-variables before being able to resolve the right method.  To
integrate it in our framework we have to include type-classes into our
unification procedure.  Luckily, this is exactly what the authors of the
theorem prover Matita accomplished using what they dubbed
\emph{unification hints} \cite{asperti2009}.  Briefly, the unifier is
given ``hints'' on how to solve problems that it cannot resolve by
itself, and such hints are repeatedly tried if unification fails.

Similar to type-classes, overloaded constructors are a feature
introduced by Agda, that lets the user define multiple data constructors
with the same name.  When such an overloaded constructor is used, its
type must be determined by the type we are type checking it against.  It
is easy to see how this problem is essentially the same as resolving the
right type class instance when encountering one of its methods---the
type-class being ``data types with the same constructors'', the
instances being the data-types, and the methods being the
constructors---, and thus we plan to implement this feature using
unification hints as well.

\subsection{Performance}

One reason for concern is that our algorithm generates more constraints
than ordinary type-checking algorithms.  However, as already remarked,
our algorithm generates a number of constraints and meta-variables which
is linear in the size of the input term.  Moreover, we would expect the
unifier to spend little time solving the trivial constraints which are
normally handled by the type-checker directly.

In the examples we have collected, we have found that this is the case,
and the run time is dominated by unification filling in implicit
arguments which can become very large.  More specifically, most of the
time is spent dealing with \ensuremath{\eta }-expansion, which we plan to improve in
the future.

\section{Conclusion}

We have presented an algorithm that leverages the unifier, an important
part of already existing dependently typed programming languages and
theorem provers, to greatly simplify the process of type checking.  The
expressivity of higher-order unification lets us specify the type-rules
concisely.  Moreover, we have clearly separated type checking from
unification, allowing for greater modularity.

We have implemented the ideas presented in the \texttt{tog}, covering a
large subset of the Agda languages.  We are currently in the process of
improving \texttt{tog} to get narrow the gap between its capabilities
and Agda.

\subsection{Acknowledgements and related work}
\label{related-work}

This work is a continuation of the work by Norell \& Coquand
\cite{norell2007b}, which describes how Agda deals with issues that we
presented.  In fact, the algorithm described here is a simpler
re-implementation of what they specified.

The other main inspiration came from a discussion with Adam Gundry about
how Epigram 2 deals with type checking in the presence of
meta-variables.  The propositional equality Epigram 2 is powerful enough
to represent constraints stuck on some meta-variable as an unfinished
equality proof.  Thus, the unifier, when given a constraint \ensuremath{\Gamma \vdash \Varid{t}\mathbin{:}\Conid{A}\mathrel{=}\Varid{u}\mathbin{:}\Conid{B}} produces a---possibly unfinished---proof that \ensuremath{\Gamma \vdash \Varid{t}\mathbin{:}\Conid{A}\equiv \Varid{u}\mathbin{:}\Conid{B}}, where here \ensuremath{\equiv } denotes Epigram's heterogeneous
propositional equality.  This proof can be used to ``transport'' terms
on one side of the equation to the other.  From this discussion I
realised that such a powerful equality was not needed to implement a
similar system. \mytodo{Insert blog citation}

The already cited work on Matita \cite{asperti2009} convinced us further
that isolating ``dynamic'' procedures in the unifier is a good idea.

Finally, our work follows a long tradition of separating elaboration of
the user syntax into a core type theory, and the type-checking of such
theory.  As far as we know this line of work goes back at least to
Coquand and Huet's ``Constructive Engine'' \cite{huet1989}, with its
separation between the ``concrete'' user syntax and the internal
``abstract'' syntax.

Moreover, we'd like to thank Andrea Vezzosi,
Nils Anders Danielsson, Conor McBride, and Bas Spitters for the useful input.

\appendix

\section{Syntax and types used in examples}
\label{examples-syntax}

Throughout the examples we use a syntax close to the one in Agda and
Haskell.  Some details regarding the syntax:
\begin{itemize}
\item type and data constructors are shown in a $\mathsf{sans}$ $\mathsf{serif}$
  font;
\item new meta-variables are introduced using underscores (\ensuremath{\anonymous });
\item binary operators are referred to with parentheses, e.g. \ensuremath{(\mathbin{+})};
\item implicit arguments are indicated by curly braces;
\item definitions can be defined by pattern matching.
\end{itemize}
We use some standard types, namely:
\begin{itemize}
\item \ensuremath{\mathsf{Set} \mathbin{:}\mathsf{Set} }---\ensuremath{\mathsf{Set} } is the type of all types, and its type is
  \ensuremath{\mathsf{Set} } itself;
\item \ensuremath{\mathsf{Bool} \mathbin{:}\mathsf{Set} }, inhabited by \ensuremath{\mathsf{true} } and \ensuremath{\mathsf{false} };
\item \ensuremath{\mathsf{Nat} \mathbin{:}\mathsf{Set} }, representing the natural numbers, introduced by
  number literals (\ensuremath{\mathrm{1}}, \ensuremath{\mathrm{2}}, etc.);
\item \ensuremath{(\equiv )\mathbin{:}\{\mskip1.5mu \Conid{A}\mathbin{:}\mathsf{Set} \mskip1.5mu\}\to \Conid{A}\to \Conid{A}\to \mathsf{Set} }, the identity type; inhabited
  by \ensuremath{\mathsf{refl} \mathbin{:}\{\mskip1.5mu \Conid{A}\mathbin{:}\mathsf{Set} \mskip1.5mu\}\to \{\mskip1.5mu \Varid{x}\mathbin{:}\Conid{A}\mskip1.5mu\}\to \Varid{x}\equiv \Varid{x}}.
\end{itemize}

\section{Substitution and term elimination}
\label{substitution}

Figure~\ref{fig:substitution} show the rules to substitute a term \ensuremath{\Varid{u}}
for a head \ensuremath{\Varid{h}}, which can be a variable \ensuremath{\Varid{x}} or a meta-variable \ensuremath{\alpha }.
Substituting \ensuremath{(\Varid{n}\;\Varid{u})[\Varid{h}\mathbin{:=}\Varid{t}]} into proper neutral terms \ensuremath{\Varid{n}\;\Varid{t}} might (in
case the head of \ensuremath{\Varid{n}} is \ensuremath{\Varid{h}}) generate redexes, which are eliminated by
invocations of the term elimination judgment \ensuremath{\Varid{t}\;\Varid{u}\leadsto \Varid{v}} and \ensuremath{\mathbf{if}\;\Varid{t}\;/\Varid{x}.\Conid{A}\;\mathbf{then}\;\Varid{u}\;\mathbf{else}\;\Varid{v}\leadsto \Varid{t}}, whose rules are given respectively in
Figure~\ref{application-elimination} and Figure~\ref{bool-elimination}.  Term
elimination in turn invokes substitution when the eliminated term is a
$\lambda$-abstraction.

While we use explicit names in our rules, we do not address issues
related to \ensuremath{\alpha }-renaming here.  In our prototype substitution is
implemented using de Bruijn indices \cite{debruijn1972}.

\begin{figure}
    \[
    \inference{}{\ensuremath{\mathsf{Set} [\Varid{h}\mathbin{:=}\Varid{t}]\leadsto \mathsf{Set} }}\hfill
    \inference{}{\ensuremath{\mathsf{Bool} [\Varid{h}\mathbin{:=}\Varid{t}]\leadsto \mathsf{Bool} }}\hfill
    \inference{}{\ensuremath{\mathsf{true} [\Varid{h}\mathbin{:=}\Varid{t}]\leadsto \mathsf{true} }}\hfill
    \inference{}{\ensuremath{\mathsf{false} [\Varid{h}\mathbin{:=}\Varid{t}]\leadsto \mathsf{false} }}
    \]
    \[
    \inference{
      \ensuremath{\Conid{A}[\Varid{h}\mathbin{:=}\Varid{t}]\leadsto \Conid{A'}} & \ensuremath{\Conid{B}[\Varid{h}\mathbin{:=}\Varid{t}]\leadsto \Conid{B'}}
    }{
      \ensuremath{((\Varid{x}\mathbin{:}\Conid{A})\to \Conid{B})[\Varid{h}\mathbin{:=}\Varid{t}]\leadsto (\Varid{x}\mathbin{:}\Conid{A'})\to \Conid{B'}}
    }\hfill
    \inference{
      \ensuremath{\Varid{u}[\Varid{h}\mathbin{:=}\Varid{t}]\leadsto \Varid{u'}}
    }{
      \ensuremath{(\lambda \Varid{x}\to \Varid{u})[\Varid{h}\mathbin{:=}\Varid{t}]\leadsto \lambda \Varid{x}\to \Varid{u'}}
    }
    \]
    \[
    \inference{}{\ensuremath{\Varid{h}[\Varid{h}\mathbin{:=}\Varid{t}]\leadsto \Varid{t}}}\hfill
    \inference{\ensuremath{\Varid{h}\not\equiv \Varid{h'}}}{\ensuremath{\Varid{h'}[\Varid{h}\mathbin{:=}\Varid{t}]\leadsto \Varid{h'}}}\hfill
    \inference{
      \ensuremath{\Varid{n}[\Varid{h}\mathbin{:=}\Varid{t}]\leadsto \Varid{v}} & \ensuremath{\Varid{u}[\Varid{h}\mathbin{:=}\Varid{t}]\leadsto \Varid{u'}} &
      \ensuremath{\Varid{v}\;\Varid{u'}\leadsto \Varid{v'}}
    }{
      \ensuremath{(\Varid{n}\;\Varid{u})[\Varid{h}\mathbin{:=}\Varid{t}]\leadsto \Varid{v'}}
    }
    \]
    \[
    \inference{
      \ensuremath{\Varid{n}[\Varid{h}\mathbin{:=}\Varid{t}]\leadsto \Varid{w}} & \ensuremath{\Conid{A}[\Varid{h}\mathbin{:=}\Varid{t}]\leadsto \Conid{A'}} & \ensuremath{\Varid{u}[\Varid{h}\mathbin{:=}\Varid{t}]\leadsto \Varid{u'}} & \ensuremath{\Varid{v}[\Varid{h}\mathbin{:=}\Varid{t}]\leadsto \Varid{v'}} \\
      \ensuremath{\mathbf{if}\;\Varid{w}\;/\Varid{x}.\Conid{A'}\;\mathbf{then}\;\Varid{u'}\;\mathbf{else}\;\Varid{v'}\leadsto \Varid{w'}}
    }{
      \ensuremath{(\mathbf{if}\;\Varid{n}\;/\Varid{x}.\Conid{A}\;\mathbf{then}\;\Varid{u}\;\mathbf{else}\;\Varid{v})[\Varid{h}\mathbin{:=}\Varid{t}]\leadsto \Varid{w'}}
    }
    \]
    \caption{\boxed{\ensuremath{\Varid{u}[\Varid{h}\mathbin{:=}\Varid{t}]\leadsto \Varid{u'}}} and \boxed{\ensuremath{\Varid{n}[\Varid{h}\mathbin{:=}\Varid{t}]\leadsto \Varid{u}}}
      Hereditary substitution into canonical and neutral terms}
    \label{fig:substitution}
\end{figure}

\begin{figure}
  \[
   \hfill
   \inference{\ensuremath{\Varid{t}[\Varid{x}\mathbin{:=}\Varid{u}]\leadsto \Varid{t'}}}{
     \ensuremath{(\lambda \Varid{x}\to \Varid{t})\;\Varid{u}\leadsto \Varid{t'}}
  }
   \hfill
  \]
  \caption{\boxed{\ensuremath{\Varid{t}\;\Varid{u}\leadsto \Varid{v}}} Application elimination}
  \label{application-elimination}
\end{figure}

\begin{figure}
  \[
  \hfill
  \inference{}{
    \ensuremath{(\mathbf{if}\;\mathsf{true} \;/\Varid{x}.\Conid{A}\;\mathbf{then}\;\Varid{t}\;\mathbf{else}\;\Varid{u})\leadsto \Varid{t}}
  }\hfill
  \inference{}{
    \ensuremath{(\mathbf{if}\;\mathsf{false} \;/\Varid{x}.\Conid{A}\;\mathbf{then}\;\Varid{t}\;\mathbf{else}\;\Varid{u})\leadsto \Varid{u}}
  }
  \hfill
  \]
  \caption{\boxed{\ensuremath{\mathbf{if}\;\Varid{t}\;/\Varid{x}.\Conid{A}\;\mathbf{then}\;\Varid{u}\;\mathbf{else}\;\Varid{v}\leadsto \Varid{t}}} \ensuremath{\mathsf{Bool} } elimination}
  \label{bool-elimination}
\end{figure}

\bibliographystyle{plain}
\bibliography{type-checking-metas}

\end{document}